%% file: main.tex
\documentclass{article} 
\usepackage{pdfpages}
\usepackage[T1]{fontenc}
\usepackage[accepted]{icml2021}
\usepackage{xr-hyper,times,yk,dsfont,hyperref,xcolor,graphicx, subfigure}

\usepackage{lipsum}
\usepackage{cuted}
\usepackage[font=small,labelfont=bf]{caption}
\usepackage{booktabs}
\usepackage{booktabs}
\input{math_commands.tex}

\usepackage{hyperref}
\usepackage{url}
\usepackage{bigints}
\usepackage{mathrsfs}
\usepackage{enumitem}
  \setlist{leftmargin=*,noitemsep}

\icmltitlerunning{Outlier-Robust Optimal Transport}

\begin{document}

\twocolumn[
\icmltitle{Outlier-Robust Optimal Transport}
\icmlsetsymbol{equal}{*}

\begin{icmlauthorlist}
\icmlauthor{Debarghya Mukherjee}{to,mitibm}
\icmlauthor{Aritra Guha}{too}
\icmlauthor{Justin Solomon}{mitc,mitibm}
\icmlauthor{Yuekai Sun}{to}
\icmlauthor{Mikhail Yurochkin}{ed,mitibm}
\end{icmlauthorlist}

\icmlaffiliation{to}{Department of Statistics, University of Michigan}
\icmlaffiliation{too}{Department of Statistical Science, Duke University}
\icmlaffiliation{mitc}{MIT CSAIL}
\icmlaffiliation{mitibm}{MIT-IBM Watson AI Lab}
\icmlaffiliation{ed}{IBM Research}

\icmlcorrespondingauthor{Debarghya Mukherjee}{mdeb@umich.edu}

\icmlkeywords{Machine Learning, ICML}

\vskip 0.3in
]

\printAffiliationsAndNotice{}

%

\makeatletter
\newcommand*{\addFileDependency}[1]{
  \typeout{(#1)}
  \@addtofilelist{#1}
  \IfFileExists{#1}{}{\typeout{No file #1.}}
}
\makeatother

\newcommand*{\myexternaldocument}[1]{%
    \externaldocument{#1}%
    \addFileDependency{#1.tex}%
    \addFileDependency{#1.aux}%
}

\myexternaldocument{Appendix}


\begin{abstract}
Optimal transport (OT) measures distances between distributions in a way that depends on the geometry of the sample space. In light of recent advances in computational OT, OT distances are widely used as loss functions in machine learning. Despite their prevalence and advantages, OT loss functions can be extremely sensitive to outliers. In fact, a single adversarially-picked outlier can increase the standard $W_2$-distance arbitrarily. To address this issue, we propose an outlier-robust formulation of OT. Our formulation is convex but challenging to scale at a first glance. Our main contribution is deriving an \emph{equivalent} formulation based on cost truncation that is easy to incorporate into modern algorithms for computational OT. We demonstrate the benefits of our formulation in mean estimation problems under the Huber contamination model in simulations and outlier detection tasks on real data.
\end{abstract}

\section{Introduction}

Optimal transport (OT) is a fundamental problem in applied mathematics. In its original form \citep{monge1781Memoire}, the problem seeks the minimum-cost way to transport mass from a probability distribution $\mu$ on $\cX$ to another distribution $\nu$ on $\cX$. In its original form, Monge's problem proved hard to study, and \citet{kantorovich1942translocation} relaxed \citeauthor{monge1781Memoire}'s formulation of the optimal transport problem to
\begin{equation}
\text{OT}(\mu,\nu) \triangleq \min_{\Pi\in\cF(\mu,\nu)}\Ex_{(X_1,X_2)\sim\Pi}\big[c(X_1,X_2)\big],
\label{eq:OT-problem}
\end{equation}
where $\cF(\mu,\nu)$ is the set of couplings between $\mu$ and $\nu$ (probability distributions on $\cX\times\cX$ whose marginals are $\mu$ and $\nu$) and $c$ is a cost function. In this paper, we assume $c(x, y) \ge 0$ and $c(x, x) = 0$.
Compared to other common measure of distance between probability distributions (\eg $d$-divergences), optimal transport uniquely depends on the geometry of the sample space (through the cost function). 

Recent advancements in optimization for optimal transport \citep{cuturi2013Sinkhorn,solomon2015convolutional,genevay2016stochastic,seguy2018LargeScale} enable its broad adaptation in machine learning applications where geometry of the data is important; see \citep{peyre2018Computational} for a survey. Optimal transport has found applications in natural language processing \citep{kusner2015word,huang2016supervised,alvarez2018gromov,yurochkin2019hierarchical}, generative modeling \citep{arjovsky2017Wasserstein}, clustering \citep{ho2017multilevel}, domain adaptation \citep{courty2014domain,courty2017joint}, large-scale Bayesian modeling \citep{srivastava2018scalable}, anomaly detection \citep{tong2020fixing}, and many other domains. 


Many applications use OT as a loss in an optimization problem of the form:
\begin{equation}\textstyle
\theta\in\argmin_{\theta\in\Theta}\text{OT}(\mu_n,\nu_\theta),
\label{eq:mke}
\end{equation}
where $\{\nu_{\theta}\}_{\theta \in \Theta}$ is a collection of parametric models and $\mu_n$ is the empirical distribution of the samples. Such estimators are called \emph{minimum Kantorovich estimators (MKE)} \citep{bassetti2006minimum}.
They are popular alternatives to likelihood-based estimators, especially in generative modeling. For example, when $\text{OT}(\cdot,\cdot)$ is the Wasserstein-1 distance and $\nu_\theta$ is a generator parameterized by a neural network with weights $\theta$, \eqref{eq:mke} corresponds to the Wasserstein GAN \citep{arjovsky2017Wasserstein}. 

One drawback of optimal transport is its sensitivity to outliers. Because \emph{all} the mass in $\mu$ must be transported to $\nu$, a small fraction of outliers can have an outsized impact. For statistics and machine learning applications in which the data is corrupted or noisy, this is a major issue. For example, the poor performance of Wasserstein GANs in the presence of outliers was noted in the recent works on outlier-robust generative learning with $f$-divergence GANs \citep{chao2018robust,wu2020minimax}.
The problem of outlier-robustness in MKE has not been studied except in two recent works proposing changes to the OT formulation that are challenging to handle computationally \citep{staerman2020ot,balaji2020robust}. Our goal is to derive an outlier-robust OT formulation compatible with existing efficient computational OT methods \citep{peyre2018Computational}.


In this paper, we propose a modification of OT to address its sensitivity to outliers. Our formulation can be used as a loss in \eqref{eq:mke}, so that it is robust to a small fraction of outliers in the data. For simplicity, we consider the $\eps$-contamination model \citep{huber2009Robust}. Let $\nu_{\theta_0}$ be a member of a parametric family  $\{\nu_\theta:\theta\in\Theta\}$ and let
\[
\mu = (1-\eps)\nu_{\theta_0} + \eps \tilde \nu,
\]
where $\mu$ is the data-generating distribution, $\eps > 0$ is the fraction of outliers, and $\tilde \nu$ is the distribution of the outliers. Although the fraction of outliers is capped at $\eps$, the value of the outliers can be arbitrary, so their effect on the optimal transport problem can be arbitrarily large. We modify the problem so that it is more robust to such outliers, targeting the downstream application of learning $\theta_0$ from (samples from) $\mu$ in the $\eps$-contamination model. 

Our main contributions are as follows:
\begin{itemize}
\item We propose a robust OT formulation suitable for statistical estimation in the $\eps$-contamination model using MKE.
\item We show that our formulation is equivalent to the original OT problem with a clipped transport cost. This connection enables us to leverage the voluminous literature on computational optimal transport to develop efficient algorithm to perform MKE robust to outliers.
\item Our formulation enables a new application of optimal transport: outlier detection in data.
\end{itemize}


\section{Problem Formulation}

\subsection{Robust OT for MKE}
To promote outlier-robustness in MKE, we need to allow the corresponding OT problem to ignore outliers in the data distribution $\mu$. The $\eps$-contamination model imposes a cap on the fraction of outliers, so it is not hard to see that $\|\mu - \nu_{\theta_0}\|_{\ \mathrm{TV}} \le \eps$, where $\|\cdot\|_{\ \mathrm{TV}}$ is the total-variation norm defined as $\|\mu\|_{\text{TV}} = \int \frac{1}{2}|\mu(\mathrm{d}x)|$. This suggests we solve a TV-constrained/regularized version of \eqref{eq:mke}:
\begin{equation}
\begin{aligned}
& \min_{\theta\in\Theta,\tilde\mu} & & \text{OT}(\tilde \mu,\nu_\theta) \\
& \subjectto             & & \|\mu - \tilde\mu\|_{\ \mathrm{TV}} \le \eps;  \ \ \tilde \mu \in \mathbf{P},
\end{aligned}
\label{eq:constrained-mke}
\end{equation}
where $\mathbf{P}$ is the set of all probability measures. The constrained version, however, suffers from identification issues. It cannot distinguish between ``clean'' distributions within TV distance $\eps$ of $\nu_{\theta_0}$. To see this, fix $\theta\in\Theta$, and note that the optimal value of \eqref{eq:constrained-mke} is zero whenever $\mu$ is within $\eps$-TV distance of $\nu_\theta$. Thus \eqref{eq:constrained-mke} cannot distinguish between two parameter values $\theta_1$ and $\theta_2$ such that $\|\nu_{\theta_1} - \nu_{\theta_2}\|_{\TV} \le \eps$. This makes it unsuitable as a loss function for statistical estimation, because it cannot lead to a consistent estimator. As an alternative, its regularized counterpart does not suffer from this issue:
\begin{equation}\label{eq:tvregularized}
\min_{\substack{\theta\in\Theta \\ s: \mu + s \in \mathbf{P}}} \text{OT}(\mu+s,\nu_\theta) + \lambda\|s\|_{\ \mathrm{TV}},
\end{equation}
where $\lambda > 0$ is a regularization parameter. Note that, the constrained and Lagrangian formulations are equivalent, but the equivalence depends on the two distributions in the arguments of the Wasserstein distance. In other words, as we vary the distributions (keeping $\eps$ fixed), the equivalent $\lambda$ will change. In our formulation, we are fixing $\lambda$ and varying the distributions, so the solution paths are different, making the parameter identifiable. In the rest of this paper, we work with the TV-regularized formulation \eqref{eq:tvregularized}.

The main idea of our formulation is to allow for modifications of $\mu$, while penalizing their magnitude and ensuring that the modified $\mu$ is still a probability measure. Below we formulate this intuition in an optimization problem titled ROBOT (ROBust Optimal Transport):

\textbf{Formulation 1:}
\begin{align}
\begin{aligned}
& \min_{\pi, s} & & \iint c(x, y) \ \pi(x, y) \ dx \ dy + \lambda \|s\|_{\mathrm{TV}} \\
& \subjectto & &  \int \pi(x, y) \ dy = \mu(x) + s(x) \geq 0 \\ 
& & & \int \pi(x, y) \ dy  = \nu(y) \\
& & & \int s(\mathrm{d}x)  = 0.
\end{aligned} 
\label{eq:robot1-cts}
\end{align}
where $\pi$ is a density function on $\cX \times \cX$. The first and the last constraints ensure that $\mu + s$ is a valid probability measure, while $\lambda \|s\|_{\text{TV}}$ penalizes the amount of modifications in $\mu$. We can identify exact locations of outliers in $\mu$ by inspecting $\mu+s$, i.e.\ if $\mu(x) + s(x) = 0$, then $x$ got eliminated and is an outlier. We will use this property to propose an outlier detection method.

ROBOT, unlike classical OT, guarantees that an adversarially-picked outliers cannot increase the (robust) transport distance arbitrarily. Let $\tilde \mu = (1-\eps)\mu + \eps \mu_c$, i.e., $\tilde \mu$ is $\mu$ contaminated with outliers from $\mu_c$, and let $\nu$ be an arbitrary measure; in MKE, $\tilde \mu$ is the contaminated data and $\nu$ is the model we learn. The adversary can arbitrarily increase $\text{OT}(\tilde \mu, \nu)$ by manipulating the outlier distribution $\mu_c$. For ROBOT, we have the following bound:
\begin{theorem}
\label{thm:bound}
Let $\tilde \mu = (1-\eps)\mu + \eps \mu_c\,$ for some $\eps\in[0,1)$. Then,
\begin{align}
\label{eq:bound}
\mathrm{ROBOT}(\tilde \mu, \nu) & \le \min\left\{\mathrm{OT}(\mu, \nu) + \lambda \eps\|\mu - \mu_c\|_{\mathrm{TV}}, \right. \notag \\
& \qquad \qquad  \left. \lambda \|\tilde \mu - \nu\|_{\mathrm{TV}}, \mathrm{OT}(\tilde \mu, \nu)\right\},
\end{align}
\end{theorem}
This bound has two key takeaways: since TV norm of distributions is bounded by 1, the adversary can not increase ROBOT$(\tilde \mu, \nu)$ arbitrarily; in the absence of outliers, ROBOT is bounded by classical OT. See Appendix \ref{sec:theorem_bnd} for the proof.




\paragraph{Related work.} We note a connection between \eqref{eq:robot1-cts} and unbalanced OT (UOT) \citep{UOT,chizat2018scaling}. UOT is typically formulated by replacing the TV norm with KL$(\mu + s | \mu)$ and adding an analogous term for $\nu$. \citet{chizat2018scaling} studied entropy regularized UOT with various divergences penalizing marginal violations. Optimization problems similar to \eqref{eq:robot1-cts} have also been considered outside of ML \citep{piccoli2014generalized,liero2018optimal}. \citet{balaji2020robust} use UOT with $\chi^2$-divergence penalty on marginal violations to achieve outlier-robustness in generative modeling. Another relevant variation of OT is partial OT \citep{figalli2010optimal,caffarelli2010free}. It may also be considered for outlier-robustness but has a drawback of forcing mass destruction rather than adjusting marginals to ignore outliers when they are present. \citet{staerman2020ot} take a different path: they replace the expectation in the Wasserstein-1 dual with a median-of-means to promote robustness. 
It is unclear what is the corresponding primal problem, making their formulation hard to interpret as an optimal transport problem.

A major challenge with the aforementioned methods, including our Formulation 1, is the large scale implementation of the optimization problem. \citet{chizat2018scaling} propose a Sinkhorn-like algorithm for entropy regularized UOT, but it is not amenable to stochastic optimization. \citet{balaji2020robust} propose a stochastic optimization algorithm based on the UOT dual, but it requires two additional neural networks (total of four including dual potentials) to parameterize modified marginal distributions (i.e., $\mu + s$ and analogous one for $\nu$). Optimizing with a median-of-means in the objective function \citep{staerman2020ot} is also challenging. The key contribution of our work is a formulation \emph{equivalent} to \eqref{eq:robot1-cts}, which is \emph{easily compatible} with the large body of classical OT optimization techniques \citep{cuturi2013Sinkhorn,solomon2015convolutional,genevay2016stochastic,seguy2018LargeScale}.

\paragraph{More efficient equivalent formulation.} At a first glance, there are two issues with \eqref{eq:robot1-cts}: it appears asymmetric, and it is unclear if it can be optimized efficiently. Below we present an \emph{equivalent} formulation that is free of these issues: \\
\textbf{Formulation 2:}
\begin{equation}
\begin{aligned}
& \min_{\Pi \in \cF^+(\bbR^{d} \times \bbR^{d})} & & \bbE_{(X,Y) \sim \Pi}\left[C_{\lambda}(X, Y)\right]  \\
& \subjectto & & X \sim \mu, \ \ Y \sim \nu \,.
\end{aligned}
\label{eq:robot2-cts}
\end{equation}
where $C_{\lambda}$ is the \emph{truncated cost} function defined as $C_{\lambda}(x, y) = \min\left\{c(x, y), 2\lambda\right\}$. Looking at \eqref{eq:robot2-cts}, it is not apparent that it adds robustness to MKE, but it is symmetric, easy to combine with entropic regularization by simply truncating the cost, and benefits from stochastic optimization algorithms \citep{genevay2016stochastic,seguy2018LargeScale}. 

This formulation also has a distant relation to the idea of loss truncation for achieving robustness \citep{shen2019learning}. \citet{pele2009fast} consider the Earth Mover Distance (discrete OT) with truncated cost to achieve computational improvements; they also mentioned its potential to promote robustness against outlier noise but did not explore this direction. 

In Section \ref{sec:theory}, we establish an \emph{equivalence} between the two ROBOT formulations, \eqref{eq:robot1-cts} and \eqref{eq:robot2-cts}. This equivalence allows us to obtain an efficient algorithm based on \eqref{eq:robot2-cts} for robust MKE. We also provide a simple procedure for computing the optimal $s$ in \eqref{eq:robot1-cts} from the solution of \eqref{eq:robot2-cts}, enabling a new OT application: outlier detection. We verify the effectiveness of robust MKE and outlier detection in our experiments in Section \ref{sec:experiments}. Before presenting the equivalence proof, however, we formulate the discrete analogs of the two ROBOT formulations for their practical value.


\subsection{Discrete ROBOT formulations}
In practice, we typically encounter samples from distributions, rather then the distributions themselves. Sampling is also built into stochastic optimization. In this subsection, we present the discrete versions of the ROBOT formulations. The key detail is that, in \eqref{eq:robot1-cts}, $\mu,\nu$ and $s$ are all supported on $\mathbb{R}^d$, while in the discrete case the empirical measures $\mu_n \in \Delta^{n-1}$ and $\nu_m \in \Delta^{m-1}$ are supported on a set of points ($\Delta^r$ is the unit probability simplex in $\reals^r$). As a result, to formulate a discrete version of \eqref{eq:robot1-cts}, we need to augment $\mu_n$ and $\nu_m$ with each others' supports. To be precise, let $\mathrm{supp}(\mu_n)=\{X_1,\dots,X_n\}$ and $\mathrm{supp}(\nu_m)=\{Y_1,\dots,Y_m\}$. Define $\cC = \{Z_1, Z_2, \dots, Z_{m+n}\} = \{X_1, \dots, X_n, Y_1, \dots, Y_m\}$.
Then discrete analog of \eqref{eq:robot1-cts} is

\textbf{Formulation 1 (discrete):}
\begin{align}
\begin{aligned}
& \min_{\Pi, \bs} & & \langle C_{aug},\Pi \rangle + \lambda \left[\|s_1\|_1 + \|t_1\|_1\right] \\
& \subjectto & & \Pi 1_{m+n} =  
\begin{bmatrix}
\mu_n + s_1 \\ t_1
\end{bmatrix},\quad
\Pi^{\top}1_{m+n} =
\begin{bmatrix}
0 \\\nu_m
\end{bmatrix}
\\
& & & \Pi \succeq 0,\quad
1_{m+n}^{\top}\bs = 0,
\end{aligned} 
\label{eq:robot1-d}
\end{align}
where $C_{aug} \in \bbR^{(m+n) \times (m+n)}$ is the augmented cost function $C_{aug, i, j} = c(Z_i, Z_j)$ ($c$ is the ground cost, e.g., squared Euclidean distance), $\bs = (s_1, t_1)$ and $1_r$ is the vector all ones in $\reals^r$. The TV norm got replaced with its discrete analog, the $L_1$ norm. Similarly to its continuous counterpart, the optimization problem is harder than typical OT due to additional constraint optimization variable $\bs$ and increased cost matrix size.


The discrete analog of \eqref{eq:robot2-cts} is straightforward:

\textbf{Formulation 2 (discrete):}
\begin{align}
\begin{aligned}
& \min_{\Pi\in\reals^{n\times m}} & & \langle C_{\lambda},\Pi \rangle  \\
& \subjectto & & \Pi1_{n} = \mu_n,\quad
\Pi^{\top}1_{m} = \nu_m,\quad
\Pi \succeq 0,
\end{aligned}
\label{eq:robot2-d}
\end{align}
where $C_{\lambda,i,j} = \min\{c(X_i, Y_j), 2\lambda\}$. As in the continuous case, it is easy to adapt modern (regularized) OT solvers without any computational overhead, and formulations of  \eqref{eq:robot1-d} and \eqref{eq:robot2-d} are equivalent. It is also possible to recover $\bs$ of \eqref{eq:robot1-d} from the solution of \eqref{eq:robot2-d} to perform outlier detection.

\paragraph{Two-sided formulation.} So far we have assumed that one of the input distributions does not have outliers, which is the setting of MKE, where the clean distribution corresponds to the model we learn. In some applications, both distributions may be corrupted. To address this case, we provide an \emph{equivalent} two-sided formulation, analogous to UOT with TV norm:

\textbf{Formulation 3 (two-sided):}
\begin{equation}
\begin{split}
\begin{aligned}
\min_{\Pi, \bs_1 , \bs_2} & \ \ \ \langle C_{aug},\Pi\rangle + \lambda \left[\|s_1\|_1 + \|t_1\|_1 + \|s_2\|_1 + \|t_2\|_1\right] \\
\subjectto & \hspace{5em}\Pi 1_{m+n} = 
\begin{bmatrix}
\mu_n+s_1 \\ t_1
\end{bmatrix} \\
& \hspace{5em} \Pi^{\top}1_{m+n} = 
\begin{bmatrix}
s_2 \\ \nu_m + t_2 
\end{bmatrix}
\\
& \Pi \succeq 0,\quad 1_{m+n}^{\top}\bs_1 = 0,\quad 1_{m+n}^{\top} \bs_2 = 0
\end{aligned}
\label{eq:F_3}
\end{split}
\end{equation}
where $\bs_1 = (s_1^{\top}, t_1^{\top})^{\top}$ and $\bs_2 = (s_2^{\top}, t_2^{\top})^{\top}$. 


\section{Equivalence of the ROBOT formulations}
\label{sec:theory}
In this section, we present our main theorem, which demonstrates the equivalence between two formulations of the robust optimal transport:
\begin{theorem}
\label{thm:main_thm}
For any two measures $\mu$ and $\nu$, \text{ROBOT}$(\mu, \nu)$ has same value for both the formulations, i.e., Formulation 1 is equivalent to Formulation 2 for the discrete case. Additionally, if the (non-truncated) cost function $c(\cdot,\cdot)$ is a metric, then the equivalence of the two formulations also holds for the continuous case. Moreover, we can recover optimal coupling of one formulation from the other. 
\end{theorem}
Below we sketch the proof of this theorem and highlight some important techniques used in the proof. We focus on the discrete case as it is more intuitive and has concrete practical implications in our experiments. A complete proof can be found in Appendix \ref{sec:proofs}. Please also see Appendix \ref{sec:proof_two_sided} for the proof of equivalence between Formulations 1, 2 and 3 in the discrete case.

\subsection{Proof sketch}
In the remainder of this section, we consider the discrete case, i.e., \eqref{eq:robot1-d} for Formulation 1 (F1) and \eqref{eq:robot2-d} for Formulation 2 (F2). Suppose $\Pi^*_{2}$ is an optimal solution of 
F2. Then we construct a feasible solution $\Pi^*_1, \bs^*_1 = (s^*_1, t^*_1)$ of F1 based on $\Pi^*_2$ with the same value of the objective function as F2 and claim that $(\Pi^*_1, \bs^*_1)$ is an optimal solution. We prove the claim by contradiction: if $(\Pi^*_1, \bs^*_1)$ is not optimal, then there exists another pair $(\tilde \Pi_1, \tilde \bs_1)$ which is optimal for F1 with strictly less objective value. We then construct another feasible solution $\Pi^*_{2,new}$ of Formulation 2 which has the same objective value as of $(\tilde \Pi_1, \tilde \bs_1)$ for F1. This implies $\Pi^*_{2,new}$ has strictly less objective value for F2 than $\Pi^*_2$, which is a contradiction. 

The two main steps of this proof are (1) constructing a feasible solution of F1 starting from a feasible solution of F2 and (2) showing that the solution constructed is indeed optimal for F1. Hence step (1) gives a recipe to construct an optimal solution of F1 starting from an optimal solution of F2. We elaborate the first point in the next subsection, which has practical implications for outlier detection. The other point is more technical; interested readers may go through the proof in Appendix \ref{sec:proof_discrete}.  

\begin{algorithm}
\caption{Generating optimal solution of F1 from F2}
\label{algo:f1-f2}
\begin{algorithmic}[1]
	\STATE Start with $\Pi^*_2 \in \bbR^{n \times m}$, an optimal solution of Formulation 2. 
	\STATE Create an augmented matrix $\Pi \in \bbR^{m+n \times m+n}$ with all $0$. Divide $\Pi$ into four blocks: 
	$$
	\Pi = \begin{bmatrix}
	\underbrace{\Pi_{11}}_{n \times n} & \underbrace{\Pi_{12}}_{n \times m} \\
	\underbrace{\Pi_{21}}_{m \times n} & \underbrace{\Pi_{22}}_{m \times m}
	\end{bmatrix}
	$$
    \STATE Set $\Pi_{12} \leftarrow \Pi^*_2$ and collect all the indices $\cI = \{(i, j): C_{i, j} > 2 \lambda \}$. 
    \STATE Set $\Pi_{12}(i, j) \leftarrow 0$ for $(i, j) \in \cI$. 
    \STATE Set $\Pi_{22}(j, j) \leftarrow \sum_{i=1}^n \Pi^*_2(i, j)\mathds{1}_{(i, j) \in \cI}$ for all $1 \le j \le m$ and set $\Pi^*_1 \leftarrow \Pi$. 
    \STATE Set $s^*_1(i) = -\sum_{j=1}^m \Pi^*_2(i, j)\mathds{1}_{(i, j) \in \cI}$ for all $1 \le i \le n$.
    \STATE Set $t_1^*(j) = \Pi_{22}(j, j)$ for all $1 \le j \le m$. 
\STATE return $\Pi^*_1, s^*_1, t^*_1$.  
\end{algorithmic}		
\end{algorithm}

\subsection{Going from Formulation 2 to Formulation 1}

Let $\Pi^*_2$ (respectively $\Pi^*_1$) be an optimal solution of F2 (respectively F1). Recall that $\Pi_1^*$ has dimension $(m+n) \times (m+n)$. From the column sum constraint in F1, we need to take the first $n$ columns of $\Pi^*_1$ to be exactly $0$, whereas the last $m$ columns must sum up to $\nu_m$. For any matrix $A$, we denote by $A[(a:b) \times (c:d)]$ the submatrix consisting of rows from $a$ to $b$ and columns from $c$ to $d$. Our main idea is to put a modified version of $\Pi^*_2$ in $\Pi^*_1[(1:n) \times (n+1:m+n)]$ and make $\Pi^*_1[(n+1:m+n) \times (n+1:m+n)]$ diagonal. First we describe how to modify $\Pi^*_2$. Observe that, if for some $(i, j)$ $C_{i, j} > 2 \lambda$, we expect $X_i \in \mathrm{supp}(\mu_n)$ to be an outlier resulting in high transportation cost, which is why we truncate the cost in F2. Therefore, to get an optimal solution of F1, we make the corresponding value of optimal plan $0$ and dump the mass into the corresponding slack variable $t_1^*$ in the diagonal of the bottom right submatrix. This changes the row sum, which is taken care of by $s_1^*$. But, as we are not moving this mass outside the corresponding column, the column sum of $\Pi^*_1[(1:(m+n)):((n+1):(m+n))]$ remains same as column sum of $\Pi^*_2$, which is $\nu_n$. We summarize this procedure in Algorithm \ref{algo:f1-f2}.
\begin{figure}[ht!]
    \centering
    \includegraphics[scale=0.39]{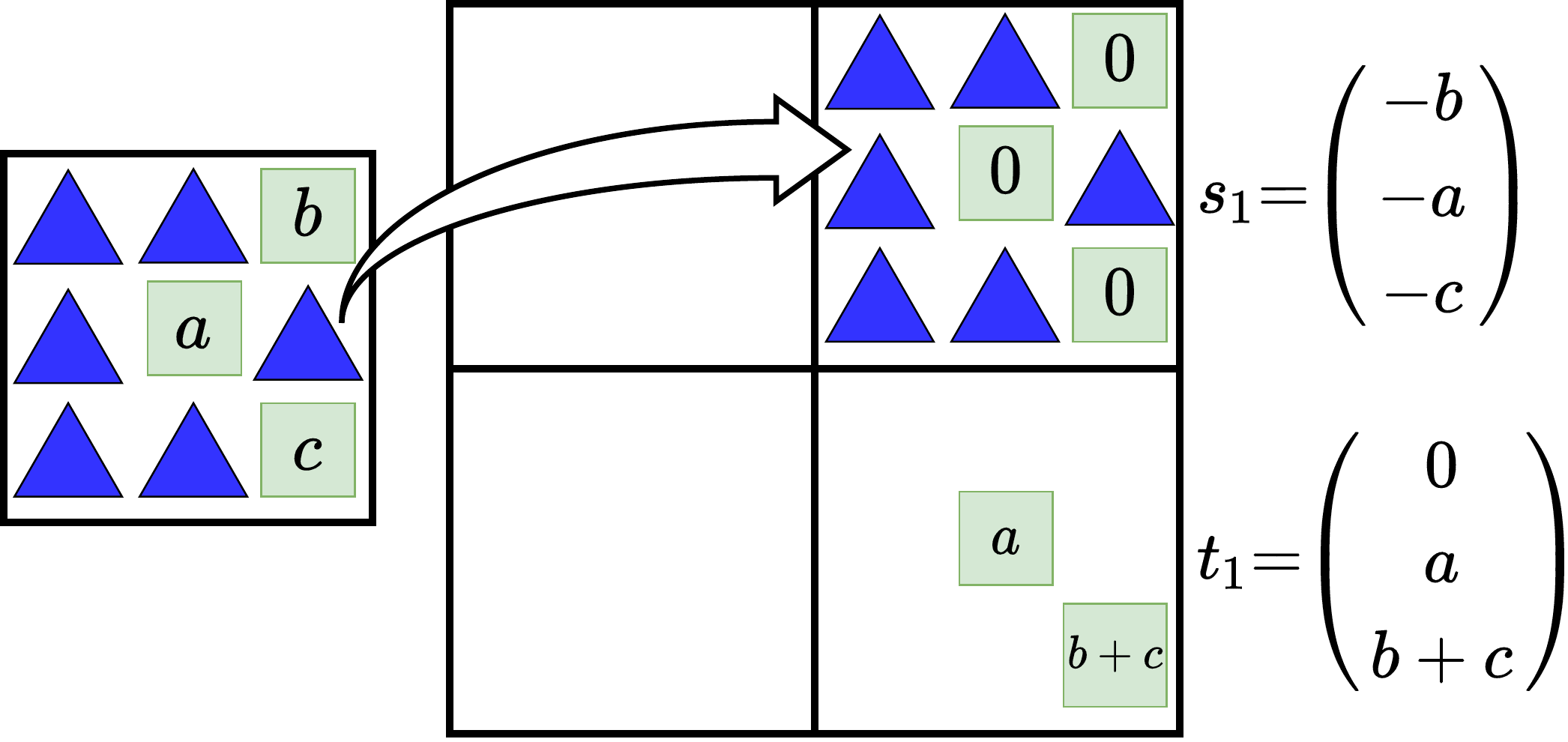}
    \caption{Constructing optimal solution of Formulation 1 from optimal solution of Formulation 2.}
    \label{fig:construction}
\end{figure}

\paragraph{Example.} In Figure \ref{fig:construction}, we provide an example to visualize the construction. On the left, we have $\Pi^*_2$, an optimal solution of Formulation 2. The blue triangles denote the positions where the corresponding cost value is $\le 2 \lambda$, and light-green squares denote the positions where the corresponding value of the cost matrix is $> 2 \lambda$. To construct an optimal solution $\Pi^*_1$ of Formulation 1 from this $\Pi^*_2$, we first create an augmented matrix of size $6 \times 6$. We keep all the entries of the left $6 \times 3$ sub-matrix as $0$ (in this picture blank elements indicate $0$). On the right submatrix, we put $\Pi^*_2$ into the top-right block, but remove the masses from light-green squares, i.e.\ where cost value is $> 2\lambda$, and put it in the diagonal entries of the bottom right block as shown in Figure \ref{fig:construction}. This mass contributes to the slack variables $s_1$ and $t_1$, and this augmented matrix along with $s_1, t_1$ give us an optimal solution of Formulation 1. 

\subsection{Outlier detection with ROBOT}
Our construction algorithm has practical consequences for outlier detection. Suppose we have two datasets, a clean dataset $\nu_m$ (i.e., has no outliers) and an outlier-contaminated dataset $\mu_n$. We can detect the outliers in $\mu_n$ without directly solving costly Formulation 1 by following Algorithm \ref{algo:outlier}. In this algorithm, $\lambda$ is a regularization parameter that can be chosen via cross-validation or heuristically (see Section \ref{sec:exp-outliers} for an example). In Section \ref{sec:exp-outliers}, we use this algorithm to perform outlier detection on image data. 
\paragraph{Outlier detection with entropic regularization.}
Algorithm \ref{algo:f1-f2} allows us to recover solution of Formulation 1, which ultimately is used for outlier detection in Algorithm \ref{algo:outlier}, by solving the simpler truncated cost Formulation 2 problem in \eqref{eq:robot2-d}. Similarly to the regular OT, it can be solved exactly with a linear program solver---\citet{pele2009fast} propose a faster exact solution based on min-cost-flow solvers benefiting from the cost truncation---or approximately using entropic regularization techniques, e.g.\ Sinkhorn algorithm \citep{cuturi2013Sinkhorn}. In the following lemma, we show that Algorithm \ref{algo:f1-f2} recovers a meaningful approximate solution of Formulation 1 from an approximate solution of Formulation 2 obtained with entropy-regularized OT solvers.

\begin{lemma}
\label{lem:entropy-f2-f1}
Let $\Pi_{2, \alpha}^*$ be a solution of the entropy regularized version of \eqref{eq:robot2-d}: 
\begin{align}
\begin{aligned}
& \argmin_{\Pi\in\reals^{n\times m}} & & \langle C_{\lambda},\Pi \rangle  + \alpha H(\Pi)\\
& \subjectto & & \Pi1_{n} = \mu_n,\quad
\Pi^{\top}1_{m} = \nu_m,\quad
\Pi \succeq 0,
\end{aligned}
\label{eq:robot2-d-entropy}
\end{align}
and let $\left(\Pi_{1, \alpha}^*, \bs^*_{1, \alpha}\right)$ be the corresponding approximate solution of Formulation 1 recovered from $\Pi^*_{2, \alpha}$ by Algorithm \ref{algo:f1-f2}. Then 
$$
\| \Pi^*_{1, \alpha} - \Pi^*_{1}\|_F + \|\bs^*_{1, \alpha} - \bs^*_{1}\|_2 \rightarrow 0
$$ 
as $\alpha \to 0$, where $\left(\Pi^*_{1}, \bs^*_{1}\right)$ is the exact solution of Formulation 1 in \eqref{eq:robot1-d}.
\end{lemma}
When the solution of \eqref{eq:robot2-d} is non-unique, then the solution of \eqref{eq:robot2-d-entropy} converges to the solution of \eqref{eq:robot2-d} with maximum entropy (see Proposition 4.1 of \citet{peyre2018Computational}) and  consequently recovers the corresponding maximum entropy solution $\left(\Pi_{1}^*, \bs^*_{1}\right)$ of \eqref{eq:robot1-d} by Algorithm \ref{algo:f1-f2} in the limit. Please find the proof in Appendix \ref{sec:lemma_sinkhorn}. 

\begin{algorithm}
\caption{Outlier detection in contaminated data}\label{algo:outlier}
\begin{algorithmic}[1]
	\STATE Start with $\mu_n$ (contaminted data) and $\nu_m$ (clean data).
	\STATE Solve Formulation 2 and obtain $\Pi^*_2$ using a suitable value of $\lambda$.
    \STATE Use Algorithm \ref{algo:f1-f2} to obtain $\Pi^*_1, s^*_1, t^*_1$ from $\Pi^*_2$.  
    \STATE Find $\cI$, the set of all the indices where $\mu_n + s_1^* = 0$.  
\STATE Return $\cI$ as the indices of outliers in $\mu_n$.  
\end{algorithmic}		
\end{algorithm}
 
\section{Empirical studies}
\label{sec:experiments}

To evaluate effectiveness of ROBOT, we consider the task of robust mean estimation under the Huber contamination model. The data is generated from $(1-\eps)\mathcal{N}(\eta_0,I_d) + \eps\mathcal{N}(\eta_1,I_d)$ and the goal is to estimate $\eta_0$. Prior work has advocated for using $f$-divergence GANs \citep{chao2018robust,wu2020minimax} for this problem and pointed out inefficiencies of Wasserstein GAN in the presence of outliers. We show that our robust OT formulation allows to estimate the uncontaminated mean $\eta_0$ comparably or better than a variety of $f$-divergence GANs. We also use this simulated setup to study sensitivity to the cost truncation hyperparameter $\lambda$.

In our second experiment, we present a new application of optimal transport enabled by ROBOT. Suppose we have collected a curated dataset $\nu_m$ (i.e., we know that it has no outliers)---such data collection is expensive, and we want to benefit from it to automate subsequent data collection. Let $\mu_n$ be a second dataset collected ``in the wild,'' i.e., it may or may not have outliers. We demonstrate how ROBOT can be used to identify outliers in $\mu_n$ using the curated dataset $\nu_m$.

\begin{table*}[t]
\caption{Robust mean estimation with GANs using different distribution divergences. True mean is $\eta_0 = \mathbf{0}_5$; sample size $n=1000$; contamination proportion $\eps=0.2$. We report results over 30 experiment restarts.}
\label{table:robogan}
\begin{center}
\begin{tabular}{lccccc}
 \toprule
 Contamination & JS Loss & SH Loss & RKL Loss & ROBOT & UOT\\
\midrule
 $\cN(0.1 \cdot \mathbf{1_5}, I_5)$   & \textbf{0.09} $\pm$ 0.03 & 0.11 $\pm$ 0.03 &   0.115 $\pm$ 0.03 & 0.1 $\pm$ 0.03 & 0.1 $\pm$ 0.04 \\
 $\cN(0.5 \cdot \mathbf{1_5}, I_5)$  &   0.23 $\pm$ 0.04 & 0.24 $\pm$ 0.05 & 0.24 $\pm$ 0.05 & \textbf{0.117} $\pm$ 0.03 & 0.2 $\pm$ 0.04 \\
 $\cN(1 \cdot \mathbf{1_5}, I_5)$ & 0.43 $\pm$ 0.05 & 0.43 $\pm$ 0.06& 0.43 $\pm$ 0.06 & 0.261 $\pm$ 0.06 & \textbf{0.25} $\pm$ 0.05 \\
 $\cN(2 \cdot \mathbf{1_5}, I_5)$ & 0.67 $\pm$ 0.07 & 0.67 $\pm$ 0.08 & 0.67 $\pm$ 0.08 & 0.106 $\pm$ 0.03 & \textbf{0.1} $\pm$ 0.03 \\
\bottomrule
\end{tabular}
\end{center}
\end{table*}


\begin{figure*}[htp]
\centering
\subfigure[Varying proportion of contamination]{\includegraphics[scale=0.31]{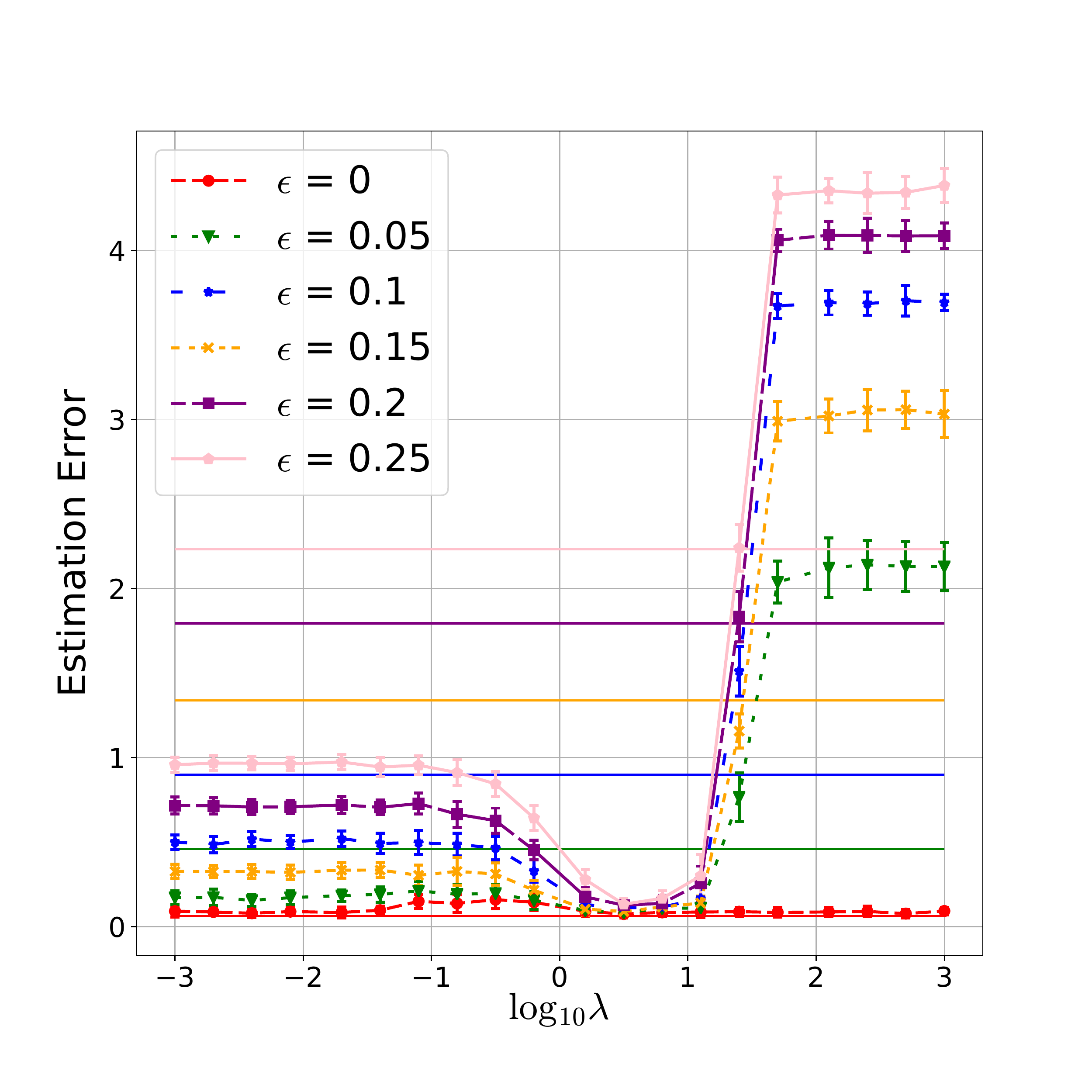}}\hspace{1cm}
  \subfigure[Varying outlier distribution mean]{\includegraphics[scale=0.31]{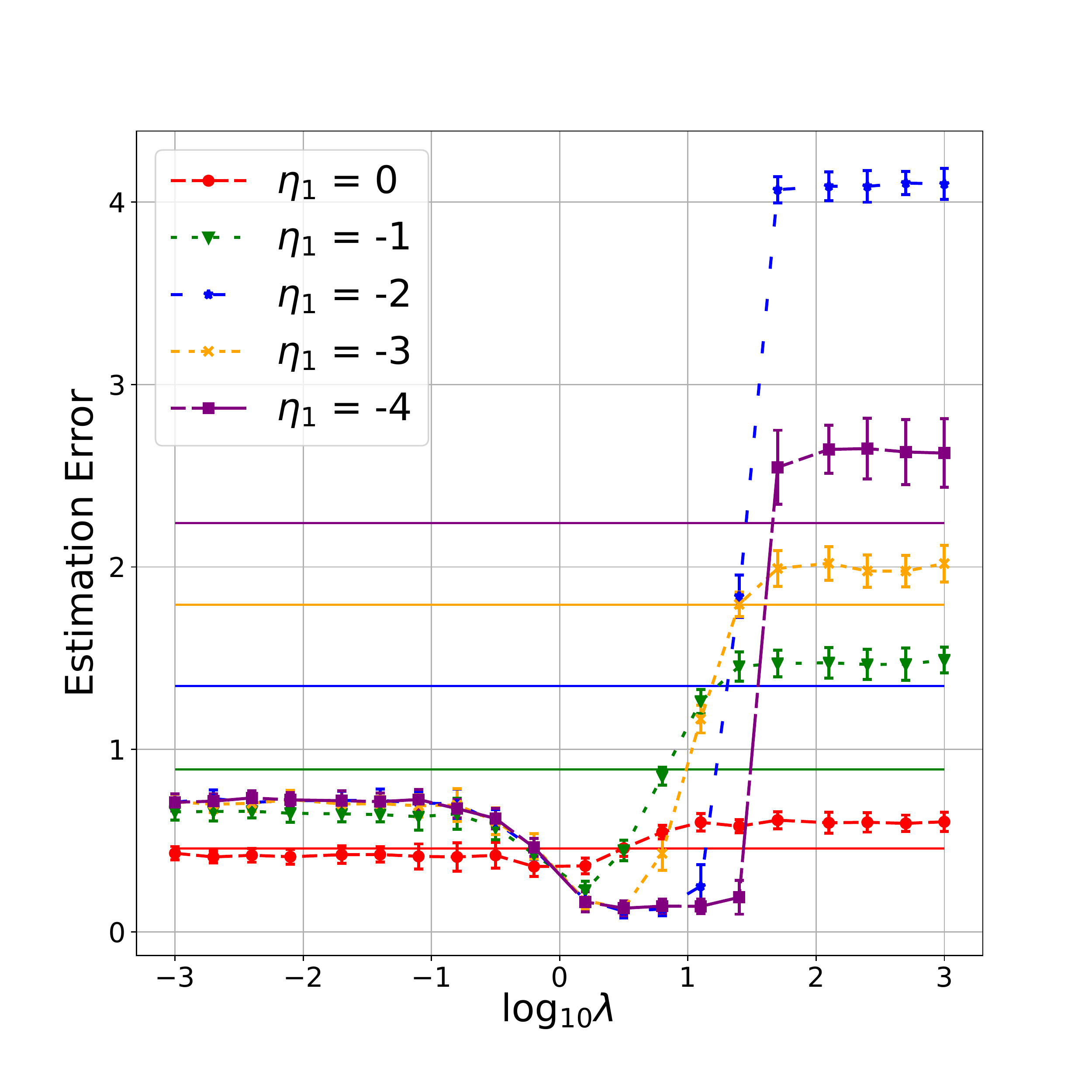}}
\caption{Empirical study of the cost truncation hyperparameter $\lambda$ sensitivity.}
\label{fig:sim-performance}
\end{figure*}

\subsection{Robust mean estimation}
\label{sec:robust_mean_est}
Following \citet{wu2020minimax}, we consider a simple generator of the form $g_\theta(x) = x + \theta$, $x \sim \mathcal{N}(0,I_d)$, $d$ is the data dimension. The basic idea of robust mean estimation with GANs is to minimize various distributional divergences between samples from $g_\theta$ and observed data simulated from $(1-\eps)\mathcal{N}(\eta_0,I_d) + \eps\mathcal{N}(\eta_1,I_d)$. The goal is to estimate $\eta_0$ with $\theta$. 

To efficiently implement ROBOT GAN, we use a standard min-max optimization approach: solve the inner max (ROBOT) and use gradient descent for the outer min parameter. To solve ROBOT, it is straightforward to adopt any of the prior stochastic regularized OT solvers: the only modification is the truncation of the cost entries as in \eqref{eq:robot2-d}. We use the stochastic algorithm for semi-discrete regularized OT \citep[Algorithm 2]{genevay2016stochastic}. We summarize ROBOT GAN in Algorithm \ref{algo:robogan}. Line 5 - Line 11 perform the inner optimization where we solve the entropy regularized OT dual with truncated cost and Line 12 - Line 14 perform gradient update of $\theta$.




For the $f$-divergence GANs \citep{nowozin2016f}, we use the code of \citet{wu2020minimax} for GANs with Jensen-Shannon (JS) loss, squared Hellinger (SH) loss, and Reverse Kullback-Leibler (RKL) loss. For the exact expressions of these divergences, see Table 1 of \citet{wu2020minimax}. We report estimation error measured by the Euclidean distance between true uncontaminated mean $\eta_0$ and estimated mean $\theta$ for various contamination distributions in Table \ref{table:robogan}. ROBOT GAN performs well across all considered contamination distributions. As the difference between true mean $\eta_0$ and contamination mean $\eta_1$ increases, the estimation error of all methods tends to increase. However, when it becomes easier to distinguish outliers from clean samples, i.e., $\eta_1 = 2 \cdot \mathbf{1_5}$, performance of ROBOT noticeably improves. We present an analogous study with data simulated from a mixture of Cauchy distributions in Appendix \ref{sec:robot_cauchy}\footnote{All codes are available at \url{https://github.com/debarghya-mukherjee/Robust-Optimal-Transport}.}.

\begin{algorithm}
\caption{ROBOT GAN}
\label{algo:robogan}
\begin{algorithmic}[1]
    \STATE \textbf{Input: } robustness regularizion $\lambda$, entropic regularization $\alpha$, data distribution $\mu_n \in \Delta^{n-1}$, $supp(\mu_n)= \cX = [X_1,\dots,X_n]$, steps sizes $\tau$ and $\gamma$
	\STATE \textbf{Initialize: } Initialize $\theta = \theta_{init}$, set number of iterations $M$ and $L$, $i = 0$, $\bv = \tilde \bv = \mathbf{0}$. 
	\FOR{$j=1,\dots,M$}
	    \STATE Generate $\tilde z \sim \cN(0, I_d)$ and set $z = \tilde z + \theta$. 
	    \STATE Set the cost vector $\bc \in \bbR^n$ as $\bc(k) = \min\{c(X_k, z), 2\lambda\}$ for $k = 1,\dots,n$.
	    \FOR{$i=1,\dots,L$} 
	    \STATE Set $\bh \leftarrow \frac{\tilde \bv - \bc}{\alpha}$ and do the normalized exponential transformation $\bu \leftarrow \frac{e^{\bh}}{\langle \mathbf{1}, e^{\bh}\rangle}$. 
	    \STATE Calculate the gradient  $\nabla \tilde \bv \leftarrow \mu_n - \bu$. 
	    \STATE Update $\tilde \bv \leftarrow \tilde \bv + \gamma \nabla \tilde \bv$ and $\bv \leftarrow (1/(j+i))\tilde \bv + (j+i-1/(j+i))\bv$.
	    \ENDFOR
	    \STATE Do the same transformation of $\bv$ as in Step 7, i.e.\ set $\bh \leftarrow \frac{\bv - \bc}{\alpha}$ and set $\Pi \leftarrow \frac{e^{\bh}}{\langle \mathbf{1}, e^{\bh}\rangle}$.
	    \STATE Set $\Pi(k) = 0$ for $k$ such that $c(X_k, z) > 2\lambda$ for $k=1,\dots,n.$
	    \STATE Calculate gradient with respect to $\theta$ as $\nabla \theta = 2\left[z\sum_{k}\Pi(k) - \cX^{\top} \Pi\right]$
    	\STATE Update $\theta \leftarrow \theta - \tau \nabla \theta$.
    \ENDFOR
    \STATE \textbf{Ouput: }$\theta$
\end{algorithmic}		
\end{algorithm}

We also compared to the Sinkhorn-based UOT algorithm \citep{chizat2018scaling} available in the Python Optimal Transport (POT) library \citep{flamary2017pot}; to obtain a UOT GAN, we modified steps 5-12 of Algorithm \ref{algo:robogan} for computing $\Pi$.
Unsurprisingly, both ROBOT and UOT perform similarly: recall equivalence to Formulation 3, which is similar to UOT with TV norm. The key insight of our work is the equivalence to classical OT with truncated cost, that greatly simplifies optimization and allows to use existing stochastic OT algorithms. In this experiment, the sample size $n=1000$ is sufficiently small for the Sinkhorn-based UOT POT implementation to be effective, but it breaks in the experiment we present in Section \ref{sec:exp-outliers}. We also tried the code of \citet{balaji2020robust} based on CVXPY \citep{diamond2016cvxpy}, but it is too slow even for the $n=1000$ sample size. In Subsection \ref{sec:balaji_comp} we present a comparison to \citet{balaji2020robust} on a smaller sample size.

\paragraph{Hyperparameter sensitivity study.} In the previous experiment, we set $\lambda = 0.5$. Now we demonstrate empirically that there is a broad range of $\lambda$ values performing well. In Figure \ref{fig:sim-performance}(a), we study sensitivity of $\lambda$ under various contamination proportions $\eps$ holding $\eta_0=\mathbf{1}_5$ and $\eta_1=5 \cdot \mathbf{1}_5$ fixed. Horizontal lines correspond to $\lambda=\infty$, i.e., vanilla OT. The key observations are: there is a wide range of $\lambda$ efficient at all contamination proportions (note the $\log_{10}$ $x$-axis scale), and ROBOT is always at least as good as vanilla OT (even when there is no contamination $\eps=0$). In Figure \ref{fig:sim-performance}(b), we present a similar study varying the mean of the contamination distribution and holding $\eps=0.2$ fixed. We see that as the contamination distribution gets closer to the true distribution, it becomes harder to pick a good $\lambda$, but the performance is always at least as good as the vanilla OT (horizontal lines).


\subsection{Outlier detection for data collection}
\label{sec:exp-outliers}
Our robust OT formulation \eqref{eq:robot1-d} enables outlier identification. Let $\nu_m$ be a clean dataset and $\mu_n$ potentially contaminated with outliers. Recall that ROBOT allows modification of one of the input distributions to eliminate potential outliers. We can identify outliers in $\mu_n$ as follows: if $\mu_n(i) + s^*_1(i) = 0$, then $X_i$, the $i$th point in $\mu_n$, is an outlier. Instead of directly solving \eqref{eq:robot1-d}, which may be inefficient, we use our equivalence results and solve an easier optimization problem \eqref{eq:robot2-d}, followed by recovering $\bs$ to find outliers via Algorithm \ref{algo:outlier}. When using entropy-regularized approximate solutions to detect outliers with Algorithm \ref{algo:outlier}, in step 4, $\mu_n + s^*_1$ will not be exactly 0 for the outliers, so a small threshold should be used instead. We modify step 4 to ``Find $\mathcal{I}$, the set of all the indices where $\mu_n + s_1^* < 1/n^2$'' when using entropy regularization.

To test our outlier-identification methodology we follow the experimental setup of \citet{tagasovska2019single}. Specifically, let $\nu_m$ be a clean dataset consisting of 10k MNIST digits from 0 to 4 and $\mu_n$ be a dataset collected ``in the wild'' consisting of (different) 8k MNIST digits from 0 to 4 and 2k outlier MNIST images of digits from 5 to 9. We compute ROBOT$(\mu_n,\nu_m)$ to identify outlier digit images in $\mu_n$. For each point in $\mu_n$ we obtain a prediction, outlier or clean, which allows us to evaluate accuracy. \citet{tagasovska2019single} use last-layer features of a neural network pre-trained on the clean data $\nu_m$---it is straightforward to combine ROBOT and other baselines we consider with any feature extractor, but in this experiment we simply use the raw data.

We compare to the Orthonormal Certificates (OC) method of \citet{tagasovska2019single} and to a variety of standard outlier detection algorithms available in Scikit-learn \citep{scikit-learn}: one class SVM \cite{scholkopf1999support}, local outlier factor \cite{breunig2000lof}, isolation forest \cite{liu2008isolation} and elliptical envelope \cite{rousseeuw1999fast}. All baselines except one class SVM and local outlier factor use clean data for training as does our method. 

\begin{table}
\centering
\caption{Outlier detection on MNIST.}
\vspace{.05in}
\begin{tabular}{lccc}
\toprule
Methods & Accuracy \\
\hline 
One class SVM & $0.496 \pm 0.003$ \\
Local outlier factor & $0.791 \pm 0.001$ \\
Isolation forest & $0.636 \pm 0.010$ \\
Elliptical envelope & $0.739 \pm 0.002$ \\
Orthonormal Certificates & $0.819 \pm 0.008$ \\
ROBOT & $0.859 \pm 0.002$ \\
ROBOT-Sinkhorn &  $\textbf{0.897} \pm 0.004$
\end{tabular}
\label{table:outliers}
\end{table}

\begin{figure}[!ht]
    \centering
    \includegraphics[scale=0.15]{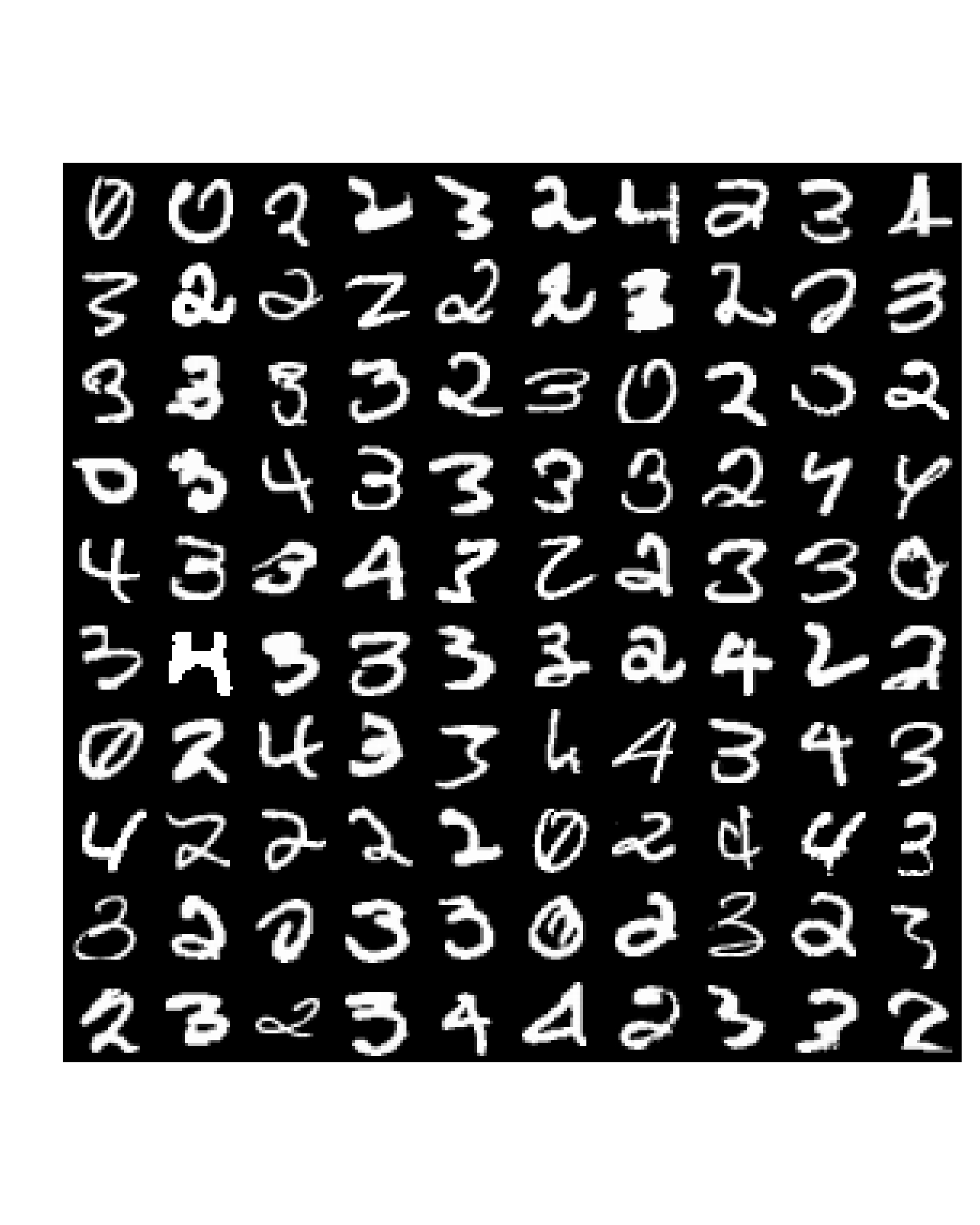}
    \qquad
    \centering 
     \includegraphics[scale=0.15]{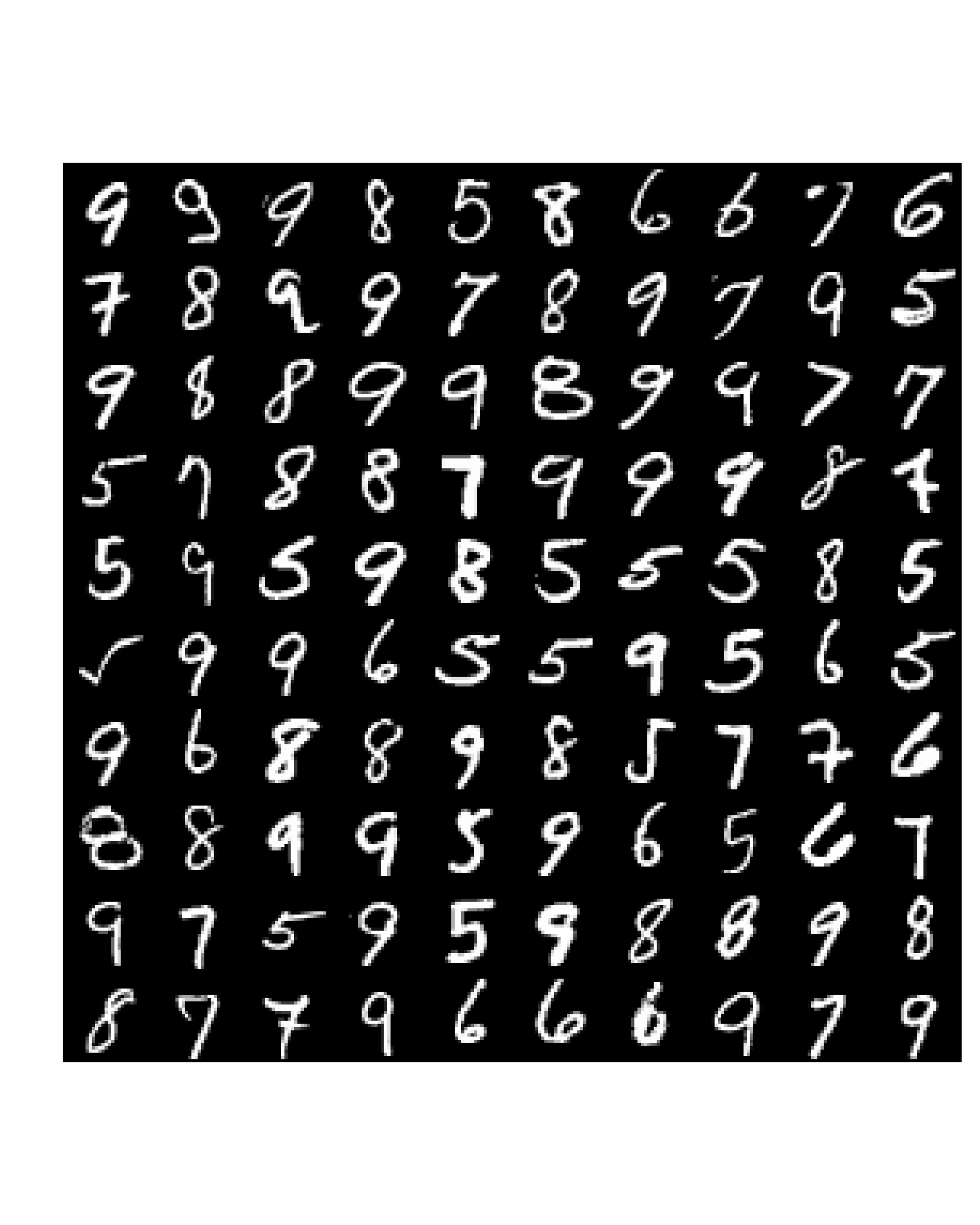}
    \qquad
     \begin{tabular}{|c|c|c|c|}
     \multicolumn{2}{c}{$\overbrace{\rule{9.5em}{0pt}}^{\text{Inliers detected as outliers}}$}& 
  \multicolumn{2}{c}{$\overbrace{\rule{9.5em}{0pt}}^{\text{Outliers detected as inliers}}$}\\
    \hline \hline
    digit & count & digit & count \\
    \hline \hline 
        0 & 10 & 5 & 191\\
        1 & 1 & 6 & 94\\
        2 & 53 & 7 & 175 \\
        3 & 32 & 8 & 181 \\
        4 & 22 & 9 & 309 \\
        \hline \hline
    \end{tabular}
    \captionlistentry[table]{A table beside a figure}
    \caption{Insights into ROBOT-Sinkhorn performance: the top left figure is a collection of random inlier digits miss-classified as outliers; the top right picture represents random outlier digits miss-classified as inliers and the table illustrates frequency distribution of miss-classified images. The majority of the errors are on digit 9, possibly due to its similarity to 4.}
    \label{fig:mnist}
\end{figure}

Results of 30 experiment repetitions are summarized in Table \ref{table:outliers}. ROBOT, i.e.\ Algorithm \ref{algo:outlier} where \eqref{eq:robot2-d} is solved exactly with a linear program solver, and ROBOT-Sinkhorn, i.e.\ Algorithm \ref{algo:outlier} where \eqref{eq:robot2-d} is solved approximately using Sinkhorn \citep{cuturi2013Sinkhorn}, produce the best results. In this experiment ROBOT-Sinkhorn outperforms ROBOT, but it is not necessarily to be expected in general. We conclude that our method is effective in assisting data collection once an initial set of clean data has been acquired and is compatible with entropy-regularized OT solvers ensuring scalability.

Elliptical envelope assumes that clean data is Gaussian, while one class SVM, isolation forest and local outlier factor correspondingly attempt to fit SVM, random forest and $k$-nearest neighbors classifiers to distinguish clean and outlier samples. All these baselines work best when the clean data is unimodal, which is not the case for the MNIST 0 to 4 digits considered in our experiment. The orthonormal certificates method is rooted in PCA and assumes that clean and outlier data live in different subspaces. This assumption is reasonable for the MNIST data, but it might not hold broadly in practice. We believe that empirical success of our method is due to its optimal transport nature. OT is a geometry-sensitive metric on distributions that can distinguish multi-modal distributions and distributions supported on the same subspace. We provide additional insights into the ROBOT performance in Figure \ref{fig:mnist}. A theoretical investigation of ROBOT outlier-detection guarantees is an interesting future work direction.


\paragraph{Hyperparameter selection.}
To select the cost truncation hyperparameter $\lambda$, we propose the following heuristic: since we know that $\nu_m$ is clean, we can subsample two datasets from it, compute vanilla OT to obtain transportation plan $\Pi$ and set $\lambda$ to be half the maximum distance between matched elements, i.e.\ $2\lambda = \max_{i,j}\{C_{ij}:\,\Pi_{ij}>0\}$, where $C$ is the cost matrix for the two subsampled datasets. This procedure is essentially estimating maximum distance between matched clean samples. To avoid subsampling noise we use 99th percentile instead of the maximum. Our experiments also revealed that increasing $\lambda$ increases the set of outliers progressively, i.e. if a sample is detected as outlier for some value of $\lambda$, then it will be detected as outlier for all higher values of $\lambda$. A rigorous theoretical analysis for this observation is a potential future work. 

The Orthonormal Certificates method \citep{tagasovska2019single} also requires setting a threshold. It computes the null-space of the clean train data and uses the norm of the projection into that space as a score to distinguish outliers. Similarly to ROBOT, and as the authors do in their code, we use 99th percentile of those scores computed on the clean train data as the threshold. For other baselines, we use the default hyperparameters.

\subsection{Comparison with \citet{balaji2020robust}}
\label{sec:balaji_comp}
We conduct additional experiments comparing to the recent robust optimal transport method of \citet{balaji2020robust}. Their method relies on CVXPY and does not scale to the sample sizes considered in our previous experiments. We compare on smaller data sizes.

\textbf{Robust mean estimation.} We set $n = 200$ samples with contamination distribution mean equal to 2 and true mean equal to $0$ (same configuration as in the last row in Table \ref{table:robogan}). The runtime and the estimation error is reported in Table \ref{table:est_balaji}.

\begin{table}
\centering
\caption{Robust mean estimation for $n =200$.}
\vspace{-.05in}
\begin{tabular}{lccc}
\toprule
Method & Estimation error & Run-time\\
\hline 
ROBOT-Sinkhorn & $0.196 \pm 0.1$ &  $35s \pm 0.6s$\\
\citet{balaji2020robust} & $0.212 \pm 0.074$ & $7745s \pm 1670s$
\end{tabular}
\label{table:est_balaji}
\end{table}

\textbf{Outlier detection experiment.} In this experiment, we consider the size of the entire dataset (inliers + outliers) to be $n=1000$ (with $800$ inliers and $200$ outliers). The results (accruacy and run-time) are provided in Table \ref{table:mnist_balaji}. 

\begin{table}
\centering
\caption{Outlier detection for $n=1000$.}
\vspace{-.05in}
\begin{tabular}{lccc}
\toprule
Method & Estimation error & Run-time \\
\hline 
ROBOT-Sinkhorn & $0.86 \pm 0.015$ & $6 \pm 2s$ \\
\citet{balaji2020robust} & $0.8 \pm 0.0005$  & $3343 \pm 960s$
\end{tabular}
\label{table:mnist_balaji}
\end{table}

The method of \citet{balaji2020robust} is significantly slower as expected. Comparing the performance, we think that ROBOT performs better because it is based on TV norm, while the method of \citet{balaji2020robust} uses a chi-squared constraints on the marginal perturbations. A TV constraint on the marginal perturbations is more closely related to the $\eps$-contamination model for outlier detection, suggesting that a TV-based constraint/regularizer could be a better choice for the outlier detection applications. We also note that in the outlier detection experiment, using chi-square divergence results in a non-sparse solution and requires tuning a threshold parameter to perform outlier detection, in addition to the chi-square distance hyperparameter. We tuned those parameters, but were not able to achieve significant performance improvements for the method of \citet{balaji2020robust}.

\section{Summary and discussion}

We propose and study ROBOT, a robust formulation of optimal transport. Although the problem is seemingly asymmetric and challenging to optimize, there is an equivalent formulation based on cost truncation that is symmetric and compatible with modern stochastic optimization methods for OT. 

ROBOT closely resembles unbalanced optimal transport (UOT). In our formulation, we added a TV regularizer to the vanilla optimal transport problem. This is motivated by the $\eps$-contamination model. In UOT, the TV regularizer is typically replaced with a KL divergence. Other choices of the regularizer may lead to new properties and applications. Studying equivalent, simpler formulations of UOT with different divergences may be a fruitful future work direction.

From the practical perspective, in our experiments we observed no degradation of ROBOT GAN in comparison to OT GAN, even when there were no outliers. It is possible that replacing OT with ROBOT may be beneficial for various machine learning applications of OT. Data encountered in practice may not be explicitly contaminated with outliers, but it often has errors and other deficiencies, suggesting that a ``no-harm'' robustness is desirable.

\section*{Acknowledgements}
This paper is based upon work supported by the National Science Foundation (NSF) under grants no. 1830247 and 1916271. J. Solomon acknowledges the generous support of Army Research Office grants W911NF1710068 and W911NF2010168, of Air Force Office of Scientific Research award FA9550-19-1-031, of National Science Foundation grant IIS-1838071, from the CSAIL Systems that Learn program, from the MIT–IBM Watson AI Laboratory, from the Toyota–CSAIL Joint Research Center, from a gift from Adobe Systems, and from the Skoltech–MIT Next Generation Program. We also thank anonymous reviewers, whose comments were extremely helpful for further improvement of our paper.

\include{Appendix_2}

\newpage
\bibliography{YK,MY,Aritra}
\bibliographystyle{iclr2021_conference}

\end{document}

%% file: math_commands.tex
\usepackage{amsmath,amsfonts,bm}

\def\eqref#1{equation~\ref{#1}}

\def\1{\bm{1}}

\def\eps{{\epsilon}}

\DeclareMathAlphabet{\mathsfit}{\encodingdefault}{\sfdefault}{m}{sl}
\SetMathAlphabet{\mathsfit}{bold}{\encodingdefault}{\sfdefault}{bx}{n}

\DeclareMathOperator*{\argmin}{arg\,min}

%% file: Appendix_2.tex
\appendix



\section{Proof of Theorem \ref{thm:main_thm}}
\label{sec:proofs}
In the proofs, $a \wedge b$ denotes $\min\{a, b\}$ for any $a, b \in \reals$. 
\subsection{Proof of discrete version }
\label{sec:proof_discrete}
\begin{proof}
Define a matrix $\Pi$ as: 
$$
\Pi(i,j) = 
\begin{cases}
0, & \text{if } C(i, j) > 2 \lambda \\
\Pi^*_2(i, j), &  \text{otherwise}
\end{cases}
$$
Also define $s \in \mathbb{R}^n$ and $t \in \mathbb{R}^m$ as: 
$$
s^*_1(i) = -\sum_{j=1}^m \Pi^*_2(i, j) \mathds{1}_{C(i, j) > 2 \lambda}
$$
and similarly define: 
$$
t^*_1(j) = \sum_{i=1}^n \Pi^*_2(i, j) \mathds{1}_{C(i, j) > 2 \lambda}
$$
These vectors corresponds to the row sums and the column sums of the elements of the optimal transport plan of Formulation 2, where the cost function exceeds $2 \lambda$. Note that, these co-ordinates of the optimal transport plan corresponding to those co-ordinates of cost matrix, where the cost is greater than $2\lambda$ and contribute to the objective value via their sum only, hence any different arrangement of these transition probabilities with same sum gives the same objective value. 

Now based on this $\Pi$ obtained we construct a feasible solution of Formulation 1 following Algorithm \ref{algo:f1-f2}: 
$$
\Pi^*_1 = 
\begin{bmatrix}
\mathbf{0} & \Pi \\
\mathbf{0} & \diag(t^*_1)
\end{bmatrix}
$$
The row sums of $\Pi^*_1$ is:
\[
\Pi^*_1 \mathbf{1}= 
\begin{bmatrix}
\mu_n + s^*_1 \\ t^*_1
\end{bmatrix}
\]
and it is immediate from the construction that the column sums of $\Pi^*_1$ is $\nu_m$. Also as: 
$$
\sum_{i=1}^n s^*_1(i) = \sum_{j=1}^m t^*_1(j) = \sum_{(i, j): C_{i, j} > 2 \lambda} \Pi^*_2(i, j) 
$$
and $s^*_1 \preceq 0, t^*_1 \succeq 0$, we have: 
$$
\mathbf{1}^{\top}(\mu_n + s^*_1 + t^*_1) = \mathbf{1}^{\top}p = 1 \,.
$$
Therefore, we have $(\Pi^*_1, s^*_1, t^*_1)$ is a feasible solution of Formulation 1. Now suppose this is not an optimal solution. Pick an optimal solution $\tilde \Pi, \tilde s, \tilde t$ of Formulation 1 so that: 
$$
\langle C_{aug}, \tilde \Pi \rangle + \lambda \left[\|\tilde s\|_1 + \|\tilde t\|_1\right]  < \langle C_{aug},  \Pi^*_1 \rangle + \lambda \left[\|s^*_1\|_1 + \|t^*_1\|_1\right] 
$$
The following two lemmas provide some structural properties of any optimal solution of Formulation 1:

\begin{lemma}
\label{lem:R1-structure}
Suppose $\Pi^*_1, s^*_1, t^*_1$ are optimal solution for Formulation 1. Divide $\Pi^*_1$ into four parts corresponding to augmentation as in algorithm \ref{algo:f1-f2}: 
$$
\Pi^*_1 = \begin{bmatrix}
\Pi^*_{1, 11} & \Pi^*_{1, 12} \\
\Pi^*_{1,21} & \Pi^*_{1, 22}
\end{bmatrix}
$$ 
Then we have $\Pi^*_{1, 11} = \Pi^*_{1, 21} = \mathbf{0}$ and $\Pi^*_{1, 22}$ is a diagonal matrix. 
\end{lemma}

\begin{lemma}
\label{lem:f2_characterization}
If $\Pi^*_1, s^*_1, t^*_1$ is an optimal solution of Formulation 1 then: 
\begin{enumerate}
    \item If $C_{i,j} > 2 \lambda$ then $\Pi^*_1(i, j) = 0$. 
    \item If $C_{i, j} < 2 \lambda$ for some $i$ and for all $1 \le j \le n$, then  $s^*_1(i) = 0$. 
    \item If $C_{i, j} < 2 \lambda$ for some $j$ and for all $1 \le i \le m$, then $t^*_1(j) = 0$. 
    \item If $C_{i, j} < 2 \lambda$ then $s^*_1(i) t^*_1(j) = 0$.
\end{enumerate}
\end{lemma}

We provide the proofs in the next subsection. By Lemma \ref{lem:R1-structure} we can assume without loss of generality: 
$$
\tilde \Pi = 
\begin{bmatrix}
\mathbf{0} & \tilde \Pi_{12} \\
\mathbf{0} & \diag(\tilde t)
\end{bmatrix}
$$
Now based on $\left(\tilde \Pi, \tilde s, \tilde t\right)$ we create a feasible solution namely $\Pi^*_{2, new}$ of Formulation 2 as follows: Define the set of indices $\{i_1, \cdots, i_k\}$ and $\{j_1, \dots, j_l\}$ as: 
$$ 
\tilde s_{i_1}, \tilde s_{i_2}, \dots, \tilde s_{i_k} > 0 \ \ \ \ \text{and} \ \ \  \tilde t_{j_1}, \tilde t_{j_2}, \dots, \tilde t_{j_l} > 0 \,. 
$$ 
Then by part (4) of Lemma \ref{lem:f2_characterization} we have $C_{i_\alpha, j_\beta} > 2 \lambda$ for $\alpha \in \{1, \dots, k\}$ and $\beta \in \{1, \dots, l\}$. Also by part (2) of Lemma \ref{lem:f2_characterization} the value of transport plan at these co-ordinates is 0. Now distribute the mass of slack variables in these co-ordinates such that the marginals of new transport plan becomes exactly $\mu_n$ and $\nu_m$. This new transport plan is our $\Pi^*_{2, new}$. Recall that, $\|\tilde s\|_1 = \| \tilde t\|_1$. Hence, here the regularizer value decreases by $2 \lambda \|\tilde s\|_1$ and the cost value increased by exactly $2 \lambda \|\tilde s\|_1$ as we are truncating the cost. Hence we have: 
\begin{align*}
    \langle C_{\lambda}, \Pi^*_{2,new} \rangle & = \langle C_{aug}, \tilde \Pi \rangle + \lambda \left[ \|\tilde s \|_1 + \|\tilde t\|_1\right] \\
    & <\langle C_{aug}, \Pi^*_1 \rangle + \lambda \left[\|s^*_1\|_1 + \|t^*_1\|_1\right] \\
    & = \langle C_{\lambda}, \Pi^*_2 \rangle
\end{align*}
which is contradiction as $\Pi^*_2$ is the optimal solution of Formulation 2. This completes the proof for the discrete part.

 \end{proof}

\subsection{Proof of equivalence for two sided formulation}
\label{sec:proof_two_sided}
Here we prove that our two sided formulation, i.e.\ Formulation 3 (\eqref{eq:F_3}) is equivalent to Formulation 1 (\eqref{eq:robot1-d}) for the discrete case. Towards that end, we introduce another auxiliary formulation and show that both Formulation 1 and Formulation 3 are equivalent to the following auxiliary formulation of the problem. 

\textbf{Formulation 4:}
\begin{equation}
\begin{aligned}
& \min\nolimits_{\Pi\in\reals^{m\times n},s_1\in\reals^m, s_2 \in \reals^n} & & \langle C,\Pi\rangle + \lambda \left[\|s_1\|_1 + \|s_2\|_1\right] \\
& \subjectto & & \Pi1_n = p + s_1 \\
& & & \Pi^T1_m = q + s_2 \\
& & & \Pi \succeq 0  
\end{aligned}. 
\label{eq:F_1}
\end{equation}
First we show that Formulation 1 and Formulation 4 are equivalent in a sense that they have the same optimal objective value. 
\begin{theorem}
\label{thm:f12}
Suppose $C$ is a cost function such that $C(x, x) = 0$. Then Formulation 1 and Formulation 4 has same optimal objective value. 
\end{theorem}
\begin{proof}
Towards that end, we show that given one optimal variables of one formulation we can get optimal variables of other formulation with the same objective value. Before going into details we need the following lemma whose proof is provided in Appendix B: 
\begin{lemma}
\label{lem:negativity}
Suppose $\Pi^*_{4}, s^*_{4, 1}, s^*_{4, 2}$ are the optimal variables of Formulation 4. Then $s^*_{4, 1} \preceq 0$ and $s^*_{4, 2} \preceq 0$. 
\end{lemma}


\noindent
Now we prove that optimal value of Formulation 1 and Formulation 4 are same. Let $(\Pi^*_{1}, s^*_{1,1}, t^*_{1,1})$ is an optimal solution of Formulation 1. Then we claim that $(\Pi^*_{1}, s^*_{1,1}, t^*_{1,1})$ is also an optimal solution of Formulation 4. Clearly it is feasible solution of Formulation 4. Suppose it is not optimal, i.e.\ there exists another optimal solution $(\tilde \Pi_{4}, \tilde s_{4, 1}, \tilde s_{4, 2})$ such that: 
$$
\langle C, \tilde \Pi_{4} \rangle + \lambda(\|\tilde s_{4, 1}\|_1 + \|\tilde s_{4, 2}\|_2) < \langle C, \Pi^*_{1, 12} \rangle + \lambda(\|s^*_{1, 1}\|_1 + \|t^*_{1, 1}\|_1)
$$
Now based on $(\tilde \Pi_{4}, \tilde s_{4, 1}, \tilde s_{4, 2})$ we construct a feasible solution of Formulation 1 as follows: 
$$
\tilde \Pi_{1} = 
\begin{bmatrix}
\mathbf{0} & \tilde \Pi_{4} \\
\mathbf{0} & -\diag(\tilde s_{4, 2})
\end{bmatrix}
$$
Note that we proved in Lemma \ref{lem:negativity} $\tilde s_{4, 2} \preceq 0$, hence we have $\tilde \Pi_{1} \succeq 0$. Now as the column sums of $\tilde \Pi_{4}$ is $q + \tilde s_{4, 2}$, we have column sums of $\tilde \Pi_{1} = [\mathbf{0} \ \ q^{\top}]^{\top}$ and the row sums are $[(p+\tilde s_{4, 1})^{\top} \ \ \ \tilde s_{4, 2}^{\top}]^{\top}$. Hence we take $\tilde s_{1, 1} = \tilde s_{4, 1}$ and $\tilde s_{1, 2} = \tilde s_{4, 2}$. 
Then it follows: 
\begin{align*}
    & \langle C_{aug}, \tilde \Pi_{1} \rangle + \lambda \left[\|\tilde s_{1 , 1}\|_1 + \|\tilde s_{1, 2} \|_1\right] \\
    & = \langle C, \tilde \Pi_{4} \rangle + \lambda \left[\|\tilde s_{4 , 1}\|_1 + \|\tilde s_{4, 2} \|_1\right] \\
    & < \langle C, \Pi^*_{1, 12} \rangle + \lambda \left[\| s^*_{1, 1}\|_1 + \|t^*_{1, 1}\|_1\right] \\
    & = \langle C_{aug}, \Pi^*_{1} \rangle + \lambda \left[\| s^*_{1, 1}\|_1 + \|t^*_{1, 1}\|_1\right]
\end{align*}
This is contradiction as we assumed $(\Pi^*_{1}, s^*_{1, 1}, t^*_{1, 2})$ is an optimal solution of Formulation 1. Therefore we conclude $(\Pi^*_{1}, s^*_{1,1}, t^*_{1,1})$ is also an optimal solution of Formulation 4 which further concludes Formulation 1 and Formulation 4 have same optimal values. This completes the proof of the theorem. 
\end{proof}
\begin{theorem}
\label{thm:f13}
The optimal objective value of Formulation 3 and Formulation 4 are same. 
\end{theorem}
\begin{proof}
Like in the proof of Theorem \ref{thm:f12}
we also prove couple of lemmas. 
\begin{lemma}
\label{lem:struct_f3}
Any optimal transport plan $\Pi^*_{3}$ of Formulation 3 has the following structure: If we write,
\[
\Pi^*_{3} = 
\begin{bmatrix}
\Pi^*_{3, 11} & \Pi^*_{3, 12} \\
\Pi^*_{3, 21} & \Pi^*_{3, 22} 
\end{bmatrix}
\]
then $\Pi^*_{3, 11}$ and $\Pi^*_{3, 22}$ are diagonal matrices and $\Pi^*_{3, 21} = \mathbf{0}$. 
\end{lemma}

\begin{lemma}
\label{lem:negativity_2}
If $s^*_{3, 1}, t^*_{3, 1},  s^*_{3, 2}, t^*_{3, 2}$ are four optimal slack variables in Formulation 3, then $s^*_{3, 1}, t^*_{3, 1} \preceq 0$ and $s^*_{3, 2}, t^*_{3, 2} \succeq 0$. 
\end{lemma}
\begin{proof}
The line of argument is same as in proof of Lemma \ref{lem:negativity}. 
\end{proof}
Next we establish equivalence. Suppose $(\Pi^*_{3}, s^*_{3, 1}, t^*_{3, 1}, s^*_{3, 2}, t^*_{3, 2})$ are optimal values of Formulation 3. We claim that $(\Pi^*_{3, 12}, s^*_{3,1} - s^*_{3,2}, t^*_{3, 1} - t^*_{3, 2})$ forms an optimal solution of Formulation 4. The objective value will then also be same as $s^*_{3, 1} \preceq 0, s^*_{3, 2} \succeq 0$ (Lemma \ref{lem:negativity_2}) implies $\|s^*_{3, 1} - s^*_{3, 2}\|_1 = \|s^*_{3, 1}\|_1 + \|s^*_{3, 2}\|_1$ and similarly $t^*_{3, 1} \preceq 0, t^*_{3, 2} \succeq 0$ implies $\|t^*_{3, 1} - t^*_{3, 2}\|_1 = \|t^*_{3, 1}\|_1 + \|t^*_{3, 2}\|_1$. Feasibility is immediate. Now for optimality, we again prove by contradiction. Suppose they are not optimal. Then lets say $\tilde \Pi_{4}, \tilde s_{4, 1}, \tilde s_{4, 2}$ are an optimal triplet of Formulation 4. Now construct another feasible solution of Formulation 3 as follows: 
Set $\tilde s_{3, 2} = \tilde t_{3, 2} = 0, \tilde s_{3, 1} = \tilde s_{4, 1} $ and $\tilde t_{3, 1} = \tilde s_{4, 2}$. Set the matrix as: 
\[
\tilde \Pi_{3} = \begin{bmatrix}
\mathbf{0} & \tilde \Pi_{4} \\
\mathbf{0} & -\diag(\tilde s_{4, 2}) 
\end{bmatrix}
\]
Then it follows that $\left(\tilde \Pi_{3}, \tilde s_{3, 1}, \tilde s_{3, 2}, \tilde t_{3, 1}, \tilde t_{3, 2}\right)$ is a feasible solution of Formulation 3. Finally we have: 
\begin{align*}
& \langle C_{aug}, \tilde \Pi_{3} \rangle + \lambda \left[\|\tilde s_{3, 1}\|_1 + \| \tilde s_{3, 2} \|_1 + \| \tilde t_{3, 1} \|_1 + \| \tilde t_{3, 2} \|_1 \right] \\
    & = \langle C_{aug}, \tilde \Pi_{3} \rangle + \lambda \left[\|\tilde s_{4, 1}\|_1 + \| \tilde s_{4, 2} \|_1 \right] \\
    & = \langle C, \tilde \Pi_{4} \rangle + \lambda \left[\|\tilde s_{4, 1}\|_1 + \| \tilde s_{4, 2} \|_1 \right] \\
    & < \langle C, \Pi^*_{3, 12} \rangle + \lambda \left[\|s^*_{3,1} - s^*_{3,2}\|_1 + \|t^*_{3, 1} - t^*_{3, 2}\|_1 \right] \\
    & = \langle C_{aug}, \Pi^*_{3} \rangle + \lambda \left[\|s^*_{3, 1}\|_1 + \|s^*_{3, 2}\|_1 + \|t^*_{3, 1}\|_1 + \|t^*_{3, 2}\|_1 \right] 
\end{align*}
This contradicts the optimality of $(\Pi^*_{3}, s^*_{3, 1}, s^*_{3, 2}, t^*_{3, 1}, t^*_{3, 2})$. This completes the proof.
\end{proof}

\subsection{Proof of continuous version}
 \begin{proof}
 
In this proof we denote by $F_1$ the optimization problem of \eqref{eq:robot1-cts} and by $F_2$ the optimization problem  \eqref{eq:robot2-cts}. Let $\mu, \nu$ be two absolutely continuous measures on $\mathbb{R}^d$. Moreover, we assume $c(x,y)=\|x-y\|$ for some norm $\|\cdot\|$ on $\mathbb{R}^d$. We assume that $\int \|x\| \nu(\mathrm{d}x), \int \|x\| \mu(\mathrm{d}x) < \infty$.

\textbf{Step 1:}  Let $K_{\epsilon}$ be a compact set such that $\int_{K_{\epsilon}} \|x\|\mu(\mathrm{d}x), \int_{K_{\epsilon}} \|x\|\nu(\mathrm{d}x) >1-\epsilon$. 

Also, let $\tilde{K}_\epsilon=\{x_1,\dots,x_{n_\epsilon}\}$ be a maximal $\epsilon$-packing set of $K_{\epsilon}$. Starting from $\tilde{K}_\epsilon$, define $\{S_1,\dots, S_{n_{\epsilon}}\}$ as a mutually disjoint covering of $K_{\epsilon}$ with internal points $x_1,\dots, x_{n_{\epsilon}}$ respectively, so that Diam$(S_i) \leq 2\epsilon$.
With $p_i=\int_{S_i} \mu(\mathrm{d}x)$, $q_i=\int_{S_i} \nu(\mathrm{d}x)$ for 
$i=1,\dots, n_{\epsilon}$, $p_0=\int_{K_{\epsilon}^C} \mu(\mathrm{d}x)$, $q_0=\int_{K_{\epsilon}^C} \nu(\mathrm{d}x)$ and $x_0=0 \in \mathbb{R}^d$, define
\begin{eqnarray}
\mu_{\epsilon}&=&\sum_0^{n_{\epsilon}}p_i \delta_{x_i} \nonumber \\
\nu_{\epsilon}&=&\sum_0^{n_{\epsilon}}q_i \delta_{x_i} \nonumber
\end{eqnarray}

A coupling $Q$ between two probability distributions is a joint distribution with marginals as the given two distributions. The Wasserstein distance between two distributions $P_1$ and $P_2$ is defined as:
\begin{eqnarray}
W_1(P_1,P_2)= \inf_{Q\in \mathscr{Q}(P_1,P_2)} \int Q(x,y)\|x-y\|\mathrm{d}x\mathrm{d}y,
\end{eqnarray}
where $\mathscr{Q}(P_1,P_2)$ is the collection of all couplings of $P_1$ and $P_2$.

Define $Q(x,y)= (\mathds{1}_{x=x_0, y \in K_{\epsilon}^C}+ \sum_{i=1}^{n_\epsilon} \mathds{1}_{x=x_i, y \in S_i})\mu(\mathrm{d}y)$. Then $Q$ is a coupling between $\mu$ and $\mu_{\epsilon}$.
 Therefore, clearly,
 \begin{eqnarray}
 W_1(\mu,\mu_{\epsilon})\leq \int_{K_{\epsilon}^C} \|x\|\mu(\mathrm{d}x) +2\epsilon \left(\sum_{i=1}^{n_{\epsilon}} p_i\right) \leq 3 \epsilon
 \end{eqnarray}
 
 Similarly, $W_1(\nu,\nu_{\epsilon}) \leq 3 \epsilon$. Therefore $\lim_{\epsilon \to 0} W_1(\nu,\nu_{\epsilon})=0$. 
 
 Moreover, $W_1(\mu,\nu)=\lim_{\epsilon \to 0} W_1(\mu_{\epsilon},\nu_{\epsilon})$, as $W_1(\mu_{\epsilon},\nu_{\epsilon})-6 \epsilon \leq W_1(\mu,\nu) \leq W_1(\mu_{\epsilon},\nu_{\epsilon})+6 \epsilon$ by triangle inequality.

\textbf{Step 2:} 
Let $S$ be an arbitrary measure with $\|S\|_{ \mathrm{TV}}= 2\gamma$, so that $\mu+S$ is a probability measure with $\int \|x\| (\mu+S) (\mathrm{d}x) <\infty$. Also, let us define $\epsilon_n=2^{-(n+1)}$. 

Let $S=S^+-S^-$, where $S^+$ and $S^-$ are positive measures on $\bbR^d$. Then, $\|S^-\|_{ \mathrm{TV}}=\|S^+\|_{ \mathrm{TV}} =\gamma $.

Clearly $(\mu-S^{-})/(1-\gamma),\mu, \nu,S^+/\gamma$ are tight probability measures. So we can construct compact sets $K_{\epsilon_n}^{(1)}$, similar to Step 1 to approximate all the four measures. Without loss of generality we assume that $0 \in K_{\epsilon_n}^{(1)}$ for all $n$. Moreover, we can also construct approximate measures $(\mu-S^-)_n=((\mu-S^-)/(1-\gamma))_{\epsilon_n}$ and $(S^+)_n= (S^+/\gamma)_{\epsilon_n}$ defined as in Step 1. $\mu_n=\mu_{\epsilon_n},\nu_n= \nu_{\epsilon_n}$ are defined similarly. All four of the measures have support points in $K_{\epsilon_n}^{(1)}$. 

Next, we define $(\mu+S)_n=\gamma(S^+)_n +(1-\gamma)(\mu-S^-)_n$. Then by the construction, from~\citep{Villani-09}, $\lim_{n \to \infty}W_1((\mu+S)_n,\mu+S) \to 0$ and thus $\lim_{n \to \infty}W_1((\mu+S)_n,\nu_n) \to W_1(\mu+S,\nu)$. Therefore we can define a signed measure $S_n=(\mu+S)_n- \mu_n$. Moreover, 
\begin{align}
\label{eq:STV}
\|S_n\|_{ \mathrm{TV}} & \leq \gamma\|(S^+)_n\|_{ \mathrm{TV}} +\|(1-\gamma)(\mu-S^-)_n-\mu_n\|_{ \mathrm{TV}} \\
& = 2\gamma= \|S\|_{ \mathrm{TV}}
\end{align}

Note that $\mu_n,\nu_n,(\mu+S)_n$ put masses (sometimes zero masses) on a common set of support points given by $\tilde{K}_{\epsilon_n}^{(1)} \subset K_{\epsilon_n}^{(1)}$. 

The $\tilde{K}_{\epsilon_n}^{(1)}$ is sequentially defined so that $\tilde{K}_{\epsilon_{n+1}}^{(1)}$ is a refinement of $\tilde{K}_{\epsilon_n}^{(1)}$. This can easily be achieved by the choice of $\epsilon_n$ defined. 
 
Consider $\tilde{s}_n,\Pi_n$ such that 
\begin{eqnarray}
\label{eq: argmin s}
F_1(\mu_n,\nu_n)=\int\|x-y\|\Pi_n(\mathrm{d}x\mathrm{d}y) +\lambda\|\tilde{s}_n\|_{ \mathrm{TV}}
\end{eqnarray}
 By the discrete nature of $\mu_n,\nu_n$, using the proof of the discrete part $F_1(\mu_n,\nu_n)= F_2(\mu_n,\nu_n)$.
 
 Since, $\min\{\|x-y\|,2\lambda\}$ is a metric, whenever $\|x-y\|$ is, therefore, it is easy to check that $F_2(\mu,\nu)=\lim_n F_2(\mu_n,\nu_n)=\lim_n F_1(\mu_n,\nu_n)$. 
 
 Moreover, by construction, $F_1(\mu_n,\nu_n) \leq \int\|x-y\|\Pi(\mathrm{d}x\mathrm{d}y) +\lambda\|S_n\|_{ \mathrm{TV}}$ for any arbitrary coupling $\Pi$ of $\mu$ and $\mu+S$. Also $\lim_n W_1(\mu_n,\mu), W_1(\mu+S,(\mu+S)_n) \to 0$. 
 
 Thus, combining the above result with \eqref{eq:STV}, we get
 \begin{eqnarray}
 \lim_n F_1(\mu_n,\nu_n) \leq \int\|x-y\|\tilde{\Pi}(\mathrm{d}x\mathrm{d}y) +\lambda\|S\|_{ \mathrm{TV}} \nonumber
 \end{eqnarray}
 for any coupling $\tilde{\Pi}$ of $\mu$ and $\mu+S$. 
 
 Therefore, $F_2(\mu,\nu) \leq F_1(\mu,\nu)$.
 

 
 

\textbf{Step 3:}  Consider $\tilde{s}_n$ defined in \eqref{eq: argmin s}. As $\tilde{s}_n$ has support in the compact sets $K_{\epsilon_n}^{(1)}$ defined in Step 2, therefore, $\{\mu_n+\tilde{s}_n\}_{n \geq 1}$ are tight measures. 

Therefore, by Prokhorov's Theorem for equivalence of sequential compactness and tightness for a collection of measures, there exists a probability measure 
$\mu \oplus s$ and a subsequence $\{n_k\}_{k \geq 1}$ such that $\mu_{n_k}+\tilde{s}_{n_k}$ converges weakly to $\mu \oplus s$. Moreover, by construction  $\lim_{R\to \infty} \limsup_{n \to \infty} \bigintss_{\|x\|>R} \|x\| (\mu_n +\nu_n)(\text{d}x) =0$ and so $\lim_{R\to \infty} \limsup_{n \to \infty} \bigintss_{\|x\|>R} \|x\|(\mu_n + \tilde{s}_n)(\text{d}x) =0$. 

Thus, by Definition 6.8 part (iii) and Theorem 6.9 of~\citep{Villani-09}, $W_1(\mu_{n_k}+\tilde{s}_{n_k}, \mu\oplus s)\to 0 $. Moreover, $W_1(\mu_{n_k},\mu)\to 0$. Therefore $\|\tilde{s}_{n_k}\|_{ \mathrm{TV}} \to \|\mu \oplus s -\mu\|_{ \mathrm{TV}}$. Thus, $W_1(\mu_{n_k}+\tilde{s}_{n_k}, \nu_{n_k}) + \lambda\|\tilde{s}_{n_k}\|_{ \mathrm{TV}} \to W_1(\mu \oplus s,\nu)+\lambda\|\mu \oplus s -\mu\|_{ \mathrm{TV}}$. But by the proof of the discrete part,
$W_1(\mu_{n_k}+\tilde{s}_{n_k}, \nu_{n_k}) + \lambda\|\tilde{s}_{n_k}\|_{ \mathrm{TV}} = F_1(\mu_{n_k},\nu_{n_k}) =F_2(\mu_{n_k},\nu_{n_k}) \to F_2(\mu, \nu)$. Therefore, with $s=\mu \oplus s -\mu$, $ W_1(\mu + s,\nu)+\lambda\|s\|_{ \mathrm{TV}}=  F_2(\mu,\nu)$. 
 
 Therefore, $F_2(\mu,\nu)=\limsup_{n \to \infty} F_1(\mu_n,\nu_n) \geq F_1(\mu,\nu)$. Thus the equality holds.

 \end{proof}

\section{Proof of Theorem \ref{thm:bound}}
\label{sec:theorem_bnd}
\begin{proof}
The proof is immediate from the Formulation 1. Recall that the Formulation 1 can restructured as: 
$$
\mathrm{ROBOT}(\tilde \mu, \nu) = \inf_{P} \left\{\mathrm{OT}(P, \nu) + \lambda \|P - \tilde \mu\|_{ \mathrm{TV}}\right\} \,.
$$
where the infimum is taking over all measure dominated by some common measure $\sigma$ (with respect to which $\mu, \mu_c, \nu$ are dominated). Hence,  
$$
\mathrm{ROBOT}(\tilde \mu, \nu) \le \mathrm{OT}(P, \nu) + \lambda \|P - \tilde \mu\|_{ \mathrm{TV}}
$$ 
for any particular choice of $P$. Taking $P = \mu$ we get that 
\begin{align*}
    \mathrm{ROBOT}(\tilde \mu, \nu) & \le \mathrm{OT}(\mu, \nu) + \lambda \|\mu - \tilde \mu\|_{ \mathrm{TV}} \\
    & = \mathrm{OT}(\mu, \nu) + \lambda \eps \|\mu - \mu_c\|_{ \mathrm{TV}}
\end{align*}
Taking $P = \nu$ we get $ \mathrm{ROBOT}(\tilde \mu, \nu) \le  \lambda \|\nu - \tilde \mu\|_{ \mathrm{TV}}$ and finally taking $P = \tilde \mu$ we get $ \mathrm{ROBOT}(\tilde \mu, \nu) \le \mathrm{OT}(\tilde \mu, \nu)$. This completes the proof. 
\end{proof}

\section{Proof of Lemma \ref{lem:entropy-f2-f1}}
\label{sec:lemma_sinkhorn}
As defined in the main text, let $\Pi^*_2$ be the optimal solution of~\eqref{eq:robot2-d} and $\Pi^*_{2, \alpha}$ be the optimal solution of \eqref{eq:robot2-d-entropy}. Then by Proposition 4.1 from ~\citet{peyre2018Computational} we conclude: 
\begin{equation}
    \label{eq:cuturi}
    \Pi^*_{2, \alpha} \overset{\alpha \to 0}{\longrightarrow} \Pi^*_2 \,.
\end{equation}
Now we have defined $\left(\Pi^*_{1, \alpha}, \bs^*_{1, \alpha}\right)$ as the \emph{approximate} solution of \eqref{eq:robot1-d} obtained via Algorithm \ref{algo:f1-f2} from $\Pi^*_{2, \alpha}$. Note that we can think of Algorithm \ref{algo:f1-f2} as a map from $ \reals^{m \times n}$ to $\reals^{(m+n) \times (m + n)} \times \reals^{m}$. Define this map as $F$. 
$$
F(\Pi_2) \mapsto (\Pi_1, \bs_1)
$$
Hence, by our notation, $\left(\Pi^*_{1, \alpha}, \bs^*_{1, \alpha}\right) = F(\Pi^*_{2, \alpha})$ and $\left(\Pi^*_{1}, \bs^*_{1}\right) = F(\Pi^*_2)$. Now if we show that $F$ is a continuous map, then by continuous mapping theorem, it is also immediate from \eqref{eq:cuturi} that: 
\begin{align*}
    F(\Pi^*_{2, \alpha}) \overset{\alpha \to 0}{\longrightarrow} F(\Pi^*_2) \,.
\end{align*}
which implies:
\begin{align*}
    \Pi^*_{1, \alpha} & \overset{\alpha \to 0}{\longrightarrow} \Pi^*_1 \\ 
    \bs^*_{1, \alpha} & \overset{\alpha \to 0}{\longrightarrow} \bs^*_1 \,. 
\end{align*}
which will complete the proof. Therefore all we need to show is that $F$ is a continuous map. Towards that direction, first fix a sequence of matrices $\{\bar \Pi_{2, i}\}_{i \in \bbN} \to \bar \Pi_2$. Define $F(\bar \Pi_{2, i}) = \left(\bar \Pi_{1, i}, \bar \bs_{1, i}\right)$ and $F(\bar \Pi_{2}) = \left(\bar \Pi_{1}, \bar \bs_{1}\right)$. By Step 3 - Step 5 of Algorithm \ref{algo:f1-f2}, we obtain $\bar \Pi_{1, i}$ by first setting $\bar \Pi_{1, i, 12} = \bar \Pi_{2, i}$ and for each of the columns of $\bar \Pi_{1, i, 12}$, dumping the sum of its entries for which the cost is $> 2 \lambda$ to the diagonals of $\bar \Pi_{1, i, 22}$. Also, we have all the entries of the first $n$ columns of $\bar \Pi_{1, i}$ to be $0$. In step 6 of Algorithm \ref{algo:f1-f2}, we obtain $\bs_{1, i}$ by taking the negative of the sum of the elements of each rows of $\bar \Pi_{1, i, 12}$ for which the cost is $> 2\lambda$. Note that these operations (Step 3 - Step 6 of Algorithm \ref{algo:f1-f2}) are continuous. Therefore we conclude:  
\begin{enumerate}
    \item $0 = \bar \Pi_{1, i, 11} \to \bar \Pi_{1, 11} = 0 \,.$
    \item $0 = \bar \Pi_{1, i, 21} \to \bar \Pi_{1, 21} = 0 \,.$ 
    \item $\bar \Pi_{1, i, 12} = \bar \Pi_{2, i} \odot \mathds{1}_{\cI^c} \to \bar \Pi_{2} \odot \mathds{1}_{\cI^c}  = \bar \Pi_{1, 12} \,.$
    \item \begin{align*}
        \bar \Pi_{1, i, 22} & = \diag\left(\mathbf{1}^{\top}\left(\bar \Pi_{2, i} \odot \mathds{1}_{\cI}\right)\right) \\
        & \to \diag\left(\mathbf{1}^{\top}\left(\bar \Pi_{2} \odot \mathds{1}_{\cI}\right)\right) \\
        & = \bar \Pi_{1, 22} \,.
    \end{align*}
    \item \begin{align*}
\bs_{1, i} & = -\left(\bar \Pi_{i, n} \odot \mathds{1}_{\cI}\right)\mathbf{1} \\
& \to  -\left(\bar \Pi_2 \odot \mathds{1}_{\cI}\right)\mathbf{1} = \bs_1 \,.
\end{align*}
\end{enumerate}
where $A \odot B$ denotes the Hadamard product (element-wise multiplication) between two matrices. Hence we have established: 
\begin{align*}
F(\bar \Pi_{2, i}) & = \left(\bar \Pi_{1, i}, \bar \bs_{1, i}\right)  \\
& \overset{n \to \infty}{\longrightarrow} \left(\bar \Pi_{1}, \bar \bs_{1}\right) \\
& = F(\bar \Pi_2) \,.
\end{align*}
This completes the proof of continuity of $F$.

 \section{Proof of auxiliary lemmas}
 \subsection{Proof of Lemma \ref{lem:R1-structure}}
 \begin{proof}
The fact that $\Pi^*_{1, 11} = \Pi^*_{1, 21} = \mathbf{0}$ follows from the fact that $\Pi^*_1 \succeq 0$ and $\Pi^*_1\mathbf{1} = \bQ$. To prove that $\Pi^*_{1, 22}$ is diagonal, we use the fact that the any diagonal entry the cost matrix is $0$. Now suppose $\Pi^*_{1, 22}$ is not diagonal. Then define a matrix $\hat \Pi$ as following: set $\hat \Pi_{11} = \hat \Pi_{21} = \mathbf{0}$, 
$\hat \Pi_{12} = \Pi^*_{1, 12}$ and: 
\[
\hat \Pi_{22}(i, j) = 
\begin{cases}
\sum_{k=1}^m \Pi^*_{1, 22}(k, i), & \text{if } j = i \\
0, & \text{if } j \neq i
\end{cases}
\]
Also define $\hat s = s^*_1$ and $\hat t$ as $\hat t(i) =  \hat \Pi_{22}(i, i)$. Then clearly $(\hat \Pi, \hat s, \hat t)$ is a feasible solution of Formulation 1. Note that:
$$
\|\hat t\|_1 = 1^{\top}\hat \Pi_{22} 1 = 1^{\top}\Pi^*_{1, 22} 1 = \|t^*_1\|_1 
$$
and by our construction $\langle C_{aug}, \hat \Pi \rangle < \langle C_{aug}, \Pi^*_1 \rangle$. Hence $(\hat \Pi, \hat s, \hat t)$ reduces the value of the objective function of Formulation 1 which is a contradiction. This completes the proof.   
\end{proof}

\subsection{Proof of Lemma \ref{lem:f2_characterization}}
\begin{proof}
\begin{enumerate}
    \item Suppose $\Pi^*_1(i, j) > 0$. Then dump this mass to $s^*_1(j)$ and make it $0$. In this way $\langle C_{aug}, \Pi^*_1 \rangle$ will decrease by $> 2 \lambda \Pi^*_1(i, j)$ and the regularizer value will increase by atmost $2 \lambda \Pi^*_1(i, j)$, resulting in overall reduction in the objective value, which leads to a contradiction. 
    \item Suppose each entry of $i^{th}$ row of $C$ is $< 2 \lambda$. Then if $s^*_1(i) > 0$, we can distribute this mass in the $i^{th}$ row such that, $s^*_1(i) = a_1 + a_2 + \dots + a_m$ with the condition that $t^*_1(j) \ge a_j$. Now we reduce $t^*_1$ as: 
$$
t^*_1(j) \leftarrow t^*_1(j) - a_j
$$
Hence the value $\langle C_{aug}, \Pi^*_1(i, j) \rangle$ will increase by a value $< 2\lambda s^*_1(i)$ but the value of regularizer will decrease by the value of $2 \lambda s^*_1(i)$, resulting in overall decrease in the value of objective function. 
\item Same as proof of part (2) by interchanging row and column in the argument.
\item Suppose not. Then choose $\eps < s^*_1(i) \wedge t^*_1(j)$, Add $\eps$ to $\Pi^*_1(i, j)$. Hence the cost function value $\langle C_{aug}, \Pi^*_1 \rangle$ will increase by $ < 2\lambda \eps$ but the regularizer value will decrease by $2 \lambda \eps$, resulting in overall decrease in the objective function. 
\end{enumerate}
\end{proof}

\subsection{Proof of Lemma \ref{lem:negativity}}
\begin{proof}
For the notational simplicity, we drop the subscript $4$ now as we will only deal with the solution of Formulation 4 and there will be no ambiguity. We prove the Lemma by contradiction. Suppose $s^*_{1, i} > 0$. Then we show one can come up with another solution $(\tilde \Pi, \tilde s_1, \tilde s_2)$ of Formulation 4 such that it has lower objective value. To construct this new solution, make: 
\[
\tilde s_{1, j} = 
\begin{cases}
s^*_{1, j}, & \text{if } j \neq i \\
0, & \text{if } j = i
\end{cases}
\]
Now to change the optimal transport plan, we will only change $i^{th}$ row of $\Pi^*$. We subtract $a_1, a_2, \dots, a_n \ge 0$ from $i^{th}$ column of $\Pi^*$ in such a way, such that none of the elements are negative. Hence the column sum will be change, i.e.\ the value of $\tilde s_2$ will be: 
\[
\tilde s_{2, j} = s^*_{2, j} - a_j \ \ \ \forall 1 \le j \le n \,.
\]
Now clearly from our construction: 
$$
\langle C, \tilde \Pi \rangle \le \langle C, \Pi^* \rangle 
$$
For the regularization part, note that, as we only reduced $i^{th}$ element of $s^*_1$, we have $\|\tilde s_1\|_1 = \|s^*_1\|_1 - s^*_{1, i}$. And by simple triangle inequality, 
$$\|\tilde s_2 \|_1 \le \|s^*_2 \|_1 + \|a_1\|_1 = \|s^*_2 \|_1 + s^*_{1, i}
$$
by construction $a_i$'s, as $a_i \ge 0$ and $\sum_i a_i = s^*_{1, i}$. Hence we have: 
$$
\|\tilde s_1\|_1 + \|\tilde s_2\|_1 \le \|s^*_1\|_1 - s^*_{1, i} + \|s^*_2\|_1 + s^*_{1, i} = \|s^*_1\|_1 + \|s^*_2\|_1 \,.
$$
Hence the value corresponding to regularizer will also decrease. This completes the proof. 
\end{proof}

\subsection{Proof of Lemma \ref{lem:struct_f3}}
\begin{proof}
We prove this lemma by contradiction. Suppose $\Pi^*_{3}$ does not have the structure mentioned in the statement of Lemma. Construct another transport plan for Formulation 3 $\tilde \Pi_{3}$ as follows: Keep $\tilde \Pi_{3, 12} = \Pi^*_{3, 12}$ and set $\tilde \Pi_{3, 12} = \mathbf{0}$. Construct the other parts as: 
\begin{align*}
& \hspace{-1em}\tilde \Pi_{3, 11}(i,j) = \\ 
& \begin{cases}
\sum_{k=1}^{m} \Pi^*_{3, 11}(i, k) + \sum_{k=1}^{n} \Pi^*_{3, 21}(k, i), & \text{if } i = j \\
0, & \text{if } i \neq j
\end{cases}
\end{align*}
and 
\[
\tilde \Pi_{3, 22}(i,j) = 
\begin{cases}
\sum_{k=1}^{n} \Pi^*_{3, 22}(k, i), & \text{if } i = j \\
0, & \text{if } i \neq j
\end{cases}
\]
It is immediate from the construction that: 
$$
\langle C_{aug}, \tilde \Pi_{3} \rangle \le \langle C_{aug}, \Pi^*_{3} \rangle 
$$
As for the regularization term: Note the by our construction $\tilde s_4$ will be same as $s_4^*$ as column sum of $\tilde \Pi_{3, 22}$ is same as $\Pi^*_{3, 22}$. For the other three:
$$
\tilde s_3(i) = \tilde \Pi_{3, 11}(i, i) = \sum_{k=1}^{m} \Pi^*_{3, 11}(i, k) + \sum_{k=1}^{n} \Pi^*_{3, 21}(k, i)
$$
$$
\tilde s_2(i) = \tilde \Pi_{3, 22}(i,i) = \sum_{k=1}^{n} \Pi^*_{3, 22}(k, i)
$$
and hence by construction: 
$$
\|\tilde s_2 \|_1 = \mathbf{1}^{\top} \Pi^*_{3, 22}\mathbf{1} = \|s^*_2\|_1 - \mathbf{1}^{\top} \Pi^*_{3, 21}\mathbf{1}\,.
$$
$$
\|\tilde s_3 \|_1 = \mathbf{1}^{\top} \Pi^*_{3, 11}\mathbf{1} + \mathbf{1}^{\top} \Pi^*_{3, 21}\mathbf{1} = \|s^*_3\|_1
$$
And also by our construction, $\tilde s_1 =  s^*_1 + c$ where $c = (\Pi^*_{3, 21})^{\top}\mathbf{1}$. As a consequence we have $\|c\|_1 = \mathbf{1}^{\top} \Pi^*_{3, 21}\mathbf{1}$. Then it follows: 
\begin{align*}
    \sum_{i=1}^4 \|\tilde s_i\|_1 & = \|s^*_1 + c\| + \|s^*_2\|_1 - \mathbf{1}^{\top} \Pi^*_{3, 21}\mathbf{1} + \|s^*_3\|_1 + \|s^*_4\|_1 \\
    & \le \sum_{i=1}^4 \|s^*_i\|_1 + \|c\|_1 - \mathbf{1}^{\top} \Pi^*_{3, 21}\mathbf{1} \\
    & = \sum_{i=1}^4 \|s^*_i\|_1 
\end{align*}
So the objective value is overall reduced. This contradicts the optimality of $\Pi^*_3$ which completes the proof. 
\end{proof}

\section{Change of support of outliers with respect to $\lambda$}
For any $\lambda$, define the set $\cI_{\lambda} = \{(i, j): C_{i, j} > 2\lambda\}$, i.e. $\cI_\lambda$ denotes the costs which exceeds the threshold $2\lambda$. As before we define by $C_\lambda$ to be truncated cost $C \wedge 2\lambda$. Denote by $\pi_\lambda$ to be the optimal transport plan with respect to $C_\lambda$ and the marginal measures $\mu, \nu$. Borrowing our notations from previous theorems, we define a "slack vector" $\bs_\lambda$ as:
$$
\bs_\lambda(i)= \sum_{j=1}^n  \pi_{\lambda}(i, j) \mathds{1}_{C(i, j) > 2\lambda} = \sum_{j: (i, j) \in \cI_\lambda}\pi_\lambda(i, j) \,.
$$
And we define the observation $i_0$ to be an outlier if $\bs_\lambda(i_0) > 0$. It is immediate that for any $\lambda_1 < \lambda_2$, $\cI_{\lambda_1} \supseteq \cI_{\lambda_2}$. We goal is to establish the following theorem: 

\begin{theorem}
For any $\lambda_1 < \lambda_2$, if $\bs_{\lambda_2}(i_0) > 0$, then $\bs_{\lambda_1}(i_0) > 0$, i.e. if a point is selected as outlier for larger $\lambda$, then it is also selected as outlier for smaller $\lambda$. 
\end{theorem}

\begin{proof}
Fix $\lambda_1 < \lambda_2$. Note that for any $\pi \in  \Pi(\mu, \nu)$ we have: 
\begin{align*}
    \langle C_{\lambda_2} - C_{\lambda_1}, \pi \rangle & = \sum_{(i, j) \in I_{\lambda_1} \cap \cI_{\lambda_2}^c} (C(i, j) - 2\lambda_1) \pi(i, j) \\
    & \qquad \qquad + 2(\lambda_2 - \lambda_1)\sum_{(i, j) \in \cI_{\lambda_2}} \pi(i, j) \\
    & := T_1(\pi) + T_2(\pi) \,.
\end{align*}
Now as $\pi_{\lambda_2}$ is optimal with respect to $C_{\lambda_2}$ and $\pi_{\lambda_1}$ is optimal with respect to $C_{\lambda_1}$ we have: 
\begin{align*}
    & \langle C_{\lambda_1}, \pi_{\lambda_2} \rangle + T_1(\pi_{\lambda_2}) + T_2(\pi_{\lambda_2}) \\
    & = \langle C_{\lambda_2} , \pi_{\lambda_2} \rangle \\
    & \le \langle C_{\lambda_2} , \pi_{\lambda_1} \rangle \\
    & = \langle C_{\lambda_1}, \pi_{\lambda_1} \rangle + T_1(\pi_{\lambda_1}) + T_2(\pi_{\lambda_1}) \\
    & \le \langle C_{\lambda_1}, \pi_{\lambda_2} \rangle + T_1(\pi_{\lambda_1}) + T_2(\pi_{\lambda_1})
\end{align*}
Therefore we have: 
\begin{equation}
    \label{eq:ineq_1_bound}
    T_1(\pi_{\lambda_2}) + T_2(\pi_{\lambda_2}) \le T_1(\pi_{\lambda_1}) + T_2(\pi_{\lambda_1}) \,.
\end{equation}
\textbf{But this is not enough.} Note that we can further decompose $T_1$ (and similarly $T_2$) as: 
\begin{align*}
T_{1,  i}(\pi) & = \sum_{j : (i, j) \in I_{\lambda_1} \cap \cI_{\lambda_2}^c} \left(C_{i, j} - 2\lambda_1\right) \pi(i, j) \\
T_{2, i}(\pi) &=  2(\lambda_2 - \lambda_1)\sum_{j : (i, j) \in \cI_{\lambda_2}} \pi(i ,j) \,.
\end{align*}
Hence we have: 
$$
T_1(\pi) = \sum_{i=1}^n T_{1, i}(\pi), \ \ \ T_2(\pi)= \sum_{i=1}^n T_{2, i}(\pi) \,.
$$
In \eqref{eq:ineq_1_bound} we have established that $T_1(\pi_{\lambda_2}) + T_2(\pi_{\lambda_2}) \le T_1(\pi_{\lambda_1}) + T_2(\pi_{\lambda_1}) \,.$ In addition if we can show that: 
\begin{equation}
    \label{eq:ineq_2_bound}
    T_{1, i_0}(\pi_{\lambda_1}) + T_{2, i_0}(\pi_{\lambda_1}) = 0 \implies T_{2, i_0}(\pi_{\lambda_1}) = 0  \,.
\end{equation}
holds for all $1 \le i \le n$ then we are done. This is because, suppose $\bs_{\lambda_1}(i_0) = 0$, Then $$
T_{1, i_0}(\pi_{\lambda_1}) + T_{2, i_0}(\pi_{\lambda_1}) = 0 \,.
$$
This in turn by \eqref{eq:ineq_2_bound} implies 
$$ 
T_{2,i_0}(\pi_{\lambda_2}) = 0 \,,
$$ 
i.e. $\bs_{\lambda_2}(i_0) = 0$.

\begin{lemma}
\label{lem:cost_2dim}
Suppose $C$ is a $2 \times 2$ cost matrix with all unequal cost: 
$$
C = \begin{bmatrix}
C_{11} & C_{12} \\
C_{21} & C_{22}
\end{bmatrix}
$$
If $C_{22} > 2\lambda_2$ and $C_{21} < 2 \lambda_1$, then the following two inequalities won't occur simultaneously: 
\begin{align*}
    C^{\lambda_2}_{11} + C^{\lambda_2}_{22} & \le C^{\lambda_2}_{12} + C^{\lambda_2}_{21} \,,\\
    C^{\lambda_1}_{12} + C^{\lambda_1}_{21} & \le C^{\lambda_1}_{11} + C^{\lambda_1}_{22} \,.
\end{align*}
\end{lemma}
Now suppose $i_2$ is an outlier with respect to $\lambda_2$ but not with respect to $\lambda_1$. Then there exists $j_1$ and $j_2$ ($j_1 \neq j_2$) such that: 
$$
C_{i_2, j_1} < 2\lambda_1, C_{i_2, j_2} > 2\lambda_2
$$ 
such that $\pi^{\lambda_2}_{i_2 ,j_2} > 0$ and $\pi^{\lambda_1}_{i_2, j_1} > 0$. 

\paragraph{Case 1: }Now assume that we can find $i_1 \neq i_2$ such that $\pi^{\lambda_2}_{i_1, j_1} > 0$ and $\pi^{\lambda_1}_{i_1, j_2} > 0$. 

Then $(i_1, j_2), (i_2, j_1) \in \supp(\pi^{\lambda_1})$ and $(i_1, j_1), (i_2, j_2) \in \supp(\pi^{\lambda_2})$. Hence from c-cyclical monotonicity properties of the support of the optimal transport plan we have for $\pi^{\lambda_1}$:
$$
 C^{\lambda_1}_{i_1, j_2} + C^{\lambda_1}_{i_2, j_1} \le C^{\lambda_1}_{i_1, j_1} + C^{\lambda_1}_{i_2, j_2} \,,
$$
and for $\pi^{\lambda_2}$: 
$$
C^{\lambda_2}_{i_1, j_1} + C^{\lambda_2}_{i_2, j_2} \le C^{\lambda_2}_{i_1, j_2} + C^{\lambda_2}_{i_2, j_1} \,.
$$
which is a contradiction from Lemma \ref{lem:cost_2dim}. This completes the proof. 

\paragraph{Case 2: Need to be proved} Now we need to consider the other case, there does not exist any row $i_1 \neq i_2$ such that both $\pi^{\lambda_2}_{i_1, j_1} > 0$ and $\pi^{\lambda_1}_{i_1, j_2} > 0$ occur simultaneously. This means that the columns $j_1,  j_2$ are orthogonal, i.e. $\langle \pi^{\lambda_2}_{:, j_1}, \pi^{\lambda_1}_{:, j_2}\rangle = 0$. 
\end{proof}

\subsection{Proof of Lemma \ref{lem:cost_2dim}}
As $C_{21} < 2\lambda_1$ and $C_{22} > 2\lambda_2$ we can modify the inequalities in Lemma \ref{lem:cost_2dim} as: 
\begin{align}
    \label{eq:cost_ineq_1} C^{\lambda_2}_{11} + 2\lambda_2 & \le C^{\lambda_2}_{12} + C_{21} \,,\\
    \label{eq:cost_ineq_2} C^{\lambda_1}_{12} + C_{21} & \le C^{\lambda_1}_{11} + 2\lambda_1 \,.
\end{align}
Now as we have assume $C_{21} < 2\lambda_1$, from \eqref{eq:cost_ineq_1} we obtain: 
\begin{align}
    & 2\lambda_1 >  C^{\lambda_2}_{11} -  C^{\lambda_2}_{12} + 2\lambda_2 \notag \\
    \iff & 2(\lambda_2 - \lambda_1) < C^{\lambda_2}_{12} - C^{\lambda_2}_{11} \,.
\end{align} 
Hence $C^{\lambda_2}_{12} - C^{\lambda_2}_{11} > 0$, which implies $C_{11} < C_{12}, C_{11} < 2\lambda_2$ and also both $C_{11}$ and $C_{12}$ can not lie within $(2\lambda_1, 2\lambda_2)$. We divide the rest of the proofs into four small cases: 
\paragraph{Case 1:} Assume $2\lambda_1 < C_{11} < 2\lambda_2, C_{12} > 2\lambda_2$. In this case from \eqref{eq:cost_ineq_1} we have: 
\begin{align*}
    C_{11} + 2\lambda_2 & \le 2\lambda_2 + C_{21}
\end{align*}
i.e. $C_{21} \ge C_{11}$ which is not possible as $C_{11} > 2\lambda_1$ and $C_{21} < 2\lambda_1$.

\paragraph{Case 2:} Assume $C_{11} < 2\lambda_1, C_{12} > 2\lambda_2$. Then from \eqref{eq:cost_ineq_2} we have $C_{21} \le C_{11}$ and from \eqref{eq:cost_ineq_1} we have: $C_{11} \le C_{21}$ which cannot occur simultaneously.

\paragraph{Case 3:} Assume $C_{11} < 2\lambda_1$ and $2\lambda_1 < C_{12} < 2\lambda_2$. Then from \eqref{eq:cost_ineq_1} and \eqref{eq:cost_ineq_2} we have respectively: 
\begin{align*}
    C_{11} + 2\lambda_2 & \le C_{12} + C_{21} \,,\\
    2\lambda_1 + C_{21} & \le C_{11} + 2\lambda_1
\end{align*}
From the second inequality we have $C_{21} \le C_{11}$, which putting back in the first inequality yields: 
$$
C_{11} + 2\lambda_2 \le C_{12} + C_{11} \implies C_{12} \ge 2\lambda_2
$$
which is a contradiction.

\paragraph{Case 4: } Assume $C_{11} < 2\lambda_1$ and $C_{12} < 2\lambda_1$. This form \eqref{eq:cost_ineq_1} yields: 
\begin{align}
    C_{11} + 2\lambda_2 & \le C_{12} + C_{21} \notag  \\
    \label{eq:cost_ineq_3} \implies C_{21} & \ge C_{11} - C_{12} + 2\lambda_2 \,.
\end{align}
Also from \eqref{eq:cost_ineq_2} we have: 
\begin{align}
    C_{12} + C_{21} & \le C_{11} + 2\lambda_1 \notag  \\
    \label{eq:cost_ineq_4} \implies C_{21} & \le C_{11} - C_{12} + 2\lambda_1 \,.
\end{align}
From \eqref{eq:cost_ineq_3} and \eqref{eq:cost_ineq_4} we have: 
$$
C_{11} - C_{12} + 2\lambda_2 \le C_{11} - C_{12} + 2\lambda_1
$$
i.e. $\lambda_2 \le \lambda_1$ which is a contradiction. This completes the proof.

\begin{table*}[ht!]
\caption{Robust mean estimation with GANs using different distribution divergences. True mean is $\eta_0 = \mathbf{0}_5$; sample size $n=1000$; contamination proportion $\eps=0.2$. We report results over 30 experiment restarts.}
\label{table:robogan_cauchy}
\begin{center}
\begin{tabular}{lccccc}
 \toprule
 Contamination & JS Loss & SH Loss & ROBOT\\
\midrule
 $\text{Cauchy}(0.1 \cdot \mathbf{1_5}, I_5)$   & 0.2 $\pm$ 0.06 & \textbf{0.17} $\pm$ 0.04 & \textbf{0.17} $\pm$ 0.05 \\
 $\text{Cauchy}(0.5 \cdot \mathbf{1_5}, I_5)$  &   0.3 $\pm$ 0.07 & 0.26 $\pm$ 0.05 & \textbf{0.25} $\pm$ 0.05 \\
 $\text{Cauchy}(1 \cdot \mathbf{1_5}, I_5)$ & 0.45 $\pm$ 0.14 & 0.37 $\pm$ 0.06 & \textbf{0.36} $\pm$ 0.07 \\
 $\text{Cauchy}(2 \cdot \mathbf{1_5}, I_5)$ & 0.39 $\pm$ 0.3 & 0.26 $\pm$ 0.06 & \textbf{0.2} $\pm$ 0.07 \\
\bottomrule
\end{tabular}
\end{center}
\end{table*}

\section{Robust mean experiment with Cauchy distribution}
\label{sec:robot_cauchy}
In this section we present our results corresponding to the robust mean estimation with the generative distribution $g_{\theta}(x) = x + \theta$ where $x \sim \text{Cauchy}(0, 1)$. As in Subsection \ref{sec:robust_mean_est}, we assume that we have observation $\{x_1, \dots, x_n\}$ from a contaminated distribution $(1 - \eps) \ \text{Cauchy}(\eta_0, 1) + \eps \  \text{Cauchy}(\eta_1, 1)$. For our experiments we take $\eta_0 = \mathbf{0}_5$ and vary $\eta_1 \in \left\{0.1 \cdot \mathbf{1_5}, 0.5 \cdot \mathbf{1_5}, 1 \cdot \mathbf{1_5}, 2 \cdot \mathbf{1_5} \right\}$ along wth $\eps = 0.2$. We compare our method with \citet{wu2020minimax} and results are presented in Table \ref{table:robogan_cauchy}.